\documentclass[sigconf, 10pt]{acmart}

\renewcommand\footnotetextcopyrightpermission[1]{} 
\usepackage[utf8]{inputenc}
\usepackage{amsmath}
\usepackage{mathtools}
\usepackage{amssymb}
\usepackage{listings}
\usepackage{array}
\usepackage{xspace}
\usepackage{graphicx}
\graphicspath{ {./figures/} }
\usepackage{varwidth}
\usepackage[ruled, linesnumbered]{algorithm2e}
\usepackage{enumitem}
\setlist{  
  listparindent=\parindent,
  parsep=0pt,
}

\newif\ifsubmission
\submissiontrue 

\newif\ifshepherd
\shepherdfalse

\ifshepherd
\newcommand{\shepherding}[1]{\textcolor{blue}{#1}}

\else
\newcommand{\shepherding}[1]{#1}
\fi

\ifsubmission
\newcommand{\note}[1]{}
\newcommand{\changed}[1]{#1}

\newcommand{\reviewone}[1]{#1}
\newcommand{\reviewtwo}[1]{#1}
\newcommand{\reviewthree}[1]{#1}
\newcommand{\onethree}[1]{#1}

\newcommand{\cut}[1]{}

\else

\newcommand{\note}[1]{\textcolor{blue}{\bf [NOTE: #1]}}
\newcommand{\changed}[1]{\textcolor{red}{#1}}

\newcommand{\reviewone}[1]{\textcolor{red}{#1}}
\newcommand{\reviewtwo}[1]{\textcolor{blue}{#1}}
\newcommand{\reviewthree}[1]{\textcolor{ForestGreen}{#1}}
\newcommand{\onethree}[1]{\textcolor{Purple}{#1}}

\newcommand{\cut}[1]{}
\fi

\newcommand{\zs}[1]{\textcolor{SkyBlue}{\bf [ZS: #1]}}

\newcommand{\ci}{result range\xspace}
\newcommand{\cis}{result ranges\xspace}
\newcommand{\milp}{MILP\xspace}

\long\def\comment#1{}
\newcommand{\stitle}[1]{\vspace{1ex} \noindent{\bf #1}}

\setcopyright{none}

\title{Fast and Reliable Missing Data Contingency Analysis with Predicate-Constraints}
\titlenote{A version of this paper has been accepted to SIGMOD 2020. This document is its associated technical report.}

\author{Xi Liang}
\affiliation{%
  \institution{University of Chicago}
}
\email{xiliang@uchicago.edu}

\author{Zechao Shang}
\affiliation{%
  \institution{University of Chicago}
}
\email{zcshang@uchicago.edu}

\author{Aaron J. Elmore}
\affiliation{%
  \institution{University of Chicago}
}
\email{aelmore@uchicago.edu}

\author{Sanjay Krishnan}
\affiliation{%
  \institution{University of Chicago}
}
\email{skr@uchicago.edu}

\author{Michael J. Franklin}
\affiliation{%
  \institution{University of Chicago}
}
\email{mjfranklin@uchicago.edu}

\begin{document}


\pagenumbering{arabic}
\setcounter{page}{1}

\begin{abstract}
Today, data analysts largely rely on intuition to determine whether missing or withheld rows of a dataset significantly affect their analyses.
We propose a framework that can produce automatic contingency analysis, i.e., the range of values an aggregate SQL query could take, under formal constraints describing the variation and frequency of missing data tuples.
We describe how to process \textsf{SUM}, \textsf{COUNT}, \textsf{AVG}, \textsf{MIN}, and \textsf{MAX} queries in these conditions resulting in hard error bounds with testable constraints.
We propose an optimization algorithm based on an integer program that reconciles a set of such constraints, even if they are overlapping, conflicting, or unsatisfiable, into such bounds. We also present a novel formulation of the Fractional Edge Cover problem to account for cases where constraints span multiple tables. Our experiments on 4 datasets against several statistical imputation and inference baselines show that statistical techniques can have a deceptively high error rate that is often unpredictable. In contrast, our framework offers hard bounds that are guaranteed to hold if the constraints are not violated. In spite of these hard bounds, we show competitive accuracy to statistical baselines.
\end{abstract}

\maketitle
\thispagestyle{plain}

\section{Introduction}
The data stored in a database may differ from real-world truth in terms of both completeness and content.
Such issues can arise due to data entry errors, inexact data integration, or software bugs~\citep{chu2016data}. 
As real-world data are rarely perfectly clean or complete, data scientists have to reason how potential sources of error may affect their analyses.
Communicating these error modes and quantifying the uncertainty they introduce into a particular analysis is arguably as important as timely execution~\citep{kraska2018northstar}. 

For example, suppose a data analyst has collected data from a temperature sensor over the span of several days.
She is interested in computing the number of times that the sensor exceeded a temperature threshold.
The data are stored in $10$ partitions; one of which failed to load into the database due to parsing errors.
The analyst can still run her query on the $9$ available partitions, however,  she needs to determine whether the loss of that partition
may affect her conclusions.

Today, analysts largely rely on intuition to reason about such scenarios.
The analyst in our example needs to make a judgment about whether the lost partition correlates with the attributes of interest, such as temperature, in any way.
Such intuitive judgments, while commonplace, are highly problematic because they are based on assumptions that are often not formally encoded in any code or documentation.
Simply reporting an extrapolated result does not convey any measure of confidence in how (in)accurate the result might be, and could hide the fact that some of the data were not used.

This paper defines a framework for specifying \reviewthree{beliefs about missing rows in a dataset} in a logical constraint language and an algorithm for computing a range of values an aggregate query can take under those constraints (hereafter called a \ci).\footnote{We use this term to differentiate a deterministic range with probabilistic confidence intervals.} This framework, which we call the Predicate-Constraint (PC) framework, facilitates several desirable outcomes: (1) the constraints are efficiently testable on historical data to determine whether or not they held true in the past, (2) the \ci is calculated deterministically and guaranteed to bound the results if the constraints hold true in the future, (3) the framework can reconcile interacting, overlapping, or conflicting \reviewthree{constraints} by enforcing the most restrictive ones, and (4) the framework makes no distributional assumptions about past data resembling future data other than what is specified in the constraints. 
With this framework, a data scientist can automatically produce a contingency analysis, i.e., the range of values the aggregate could take, under formally described assumptions about the nature of the unseen data.
\reviewthree{Since the assumptions are formally described and completely determine the \cis, they can be checked, versioned, and tested just like any other analysis code---ultimately facilitating a more reproducible analysis methodology.}

\reviewthree{The constraints themselves, called }Predicate-Constraints, are logical statements that constrain the range of values that a set of rows can take and the number of such rows within a predicate.
We show that deriving the \cis for a single ``closed'' predicate-constraint set can be posed as a mixed-integer linear program (\milp).
We show links to the Fractional Edge Cover bounds employed in the analysis of worst-case optimal joins when we consider constraints over multiple tables and contribute a new variant of the same problem which can be used to bound predicate-constraints~\citep{ngo2018worst}.
The solver itself contains a number of novel optimizations, which we contribute such as early pruning of unsatisfiable search paths.

To the best of our knowledge, a direct competitor framework does not exist. 
While there is a rich history of what-if analysis~\citep{deutch2013caravan} and how-to analysis~\citep{meliou2011reverse}, which characterize a database's behavior under hypothetical updates, analyzing the effects of constrained missing rows on aggregate queries has been not been extensively studied.
The closest such framework is the m-table framework~\citep{sundarmurthy2017m}, which has a similar expressiveness but no algorithm for computing aggregate \cis.
Likewise, some of the work in data privacy solves a simplified version of the problem where there are no overlapping constraints or join conditions~\citep{zhang2007aggregate}. 

In summary, we contribute:
\begin{enumerate}
\item A formal framework for contingency analysis over missing or withheld rows of data, where users specify constraints about the frequency and variation of the missing rows.
\item An optimization algorithm that reconciles a set of such constraints, even if they are overlapping, conflicting, or unsatisfiable, into a range of possible values that \textsf{SUM}, \textsf{COUNT}, \textsf{AVG}, \textsf{MIN}, and \textsf{MAX} SQL queries can take.
\item A novel formulation of the Fractional Edge Cover problem to account for cases where constraints span multiple tables and improve accuracy for natural joins.
\item Meta-optimizations that improve accuracy and/or optimization performance such as pruning unsatisfiable constraint paths.
\end{enumerate}

\section{Background}
\label{sec:background}
\shepherding{In this paper, we consider the following user interface. The system is asked to answer SQL aggregate queries over a table with a number of missing rows. The user provides a set of constraints (called predicate-constraints) that
describes how many such rows are missing and a range of possible attribute values those missing rows could take. 
The system should integrate these constraints into its query processing and compute the maximal range of results (aggregate values) consistent with those constraints.}

All of our queries are of the form:
\begin{lstlisting}
SELECT agg(attr)
FROM R1,...,RN
WHERE ....
GROUP BY ....
\end{lstlisting}
We consider \textsf{SUM}, \textsf{COUNT}, \textsf{AVG}, \textsf{MIN}, and \textsf{MAX} aggregates with predicates and possible inner join conditions. 
\shepherding{Because GROUP-BY clause can be considered as a union of such queries without GROUP-BY. In the rest of the paper, we focus on queries without GROUP-BY clause.}

\subsection{Example Application}
\label{sec:ex-application}
\onethree{Consider a simplified sales transaction table of just three attributes:}
\begin{lstlisting}
Sales(utc, branch, price)
Nov-01 10:20,New York,3.02
Nov-01 10:21,Chicago,6.71
...
Nov-16 6:42,Trenton,18.99
\end{lstlisting}
\onethree{Over this dataset, a data analyst is interested in calculating the total number of sales:}
\begin{lstlisting}
SELECT SUM(price)
FROM Order
\end{lstlisting}
\onethree{Suppose that between November 10 and November 13 there was a network outage that caused data from the New York and Chicago branches to be lost. How can we assess the effect of the missing data on the query result?}

\vspace{0.25em} \noindent \onethree{\textbf{Simple Extrapolation: } One option is to simply extrapolate the SUM for the missing days based on the data that is available. While this solution is simple to implement, it leads to subtle assumptions that can make the analysis very misleading. Extrapolation assumes that the missing data comes from roughly the same distribution as the data that is present. Figure \ref{se-teaser} shows an experiment on one of our experimental datasets. We vary the mount of missing data in a way that is correlated with a SUM query. Even if the exact amount of missing data is known, the estimate become increasingly error prone.
More subtly, extrapolation returns a single result without a good measure of uncertainty---there is no way to know how wrong an answer might be.}

\begin{figure}[t]
    \centering
    \includegraphics[width=1\linewidth]{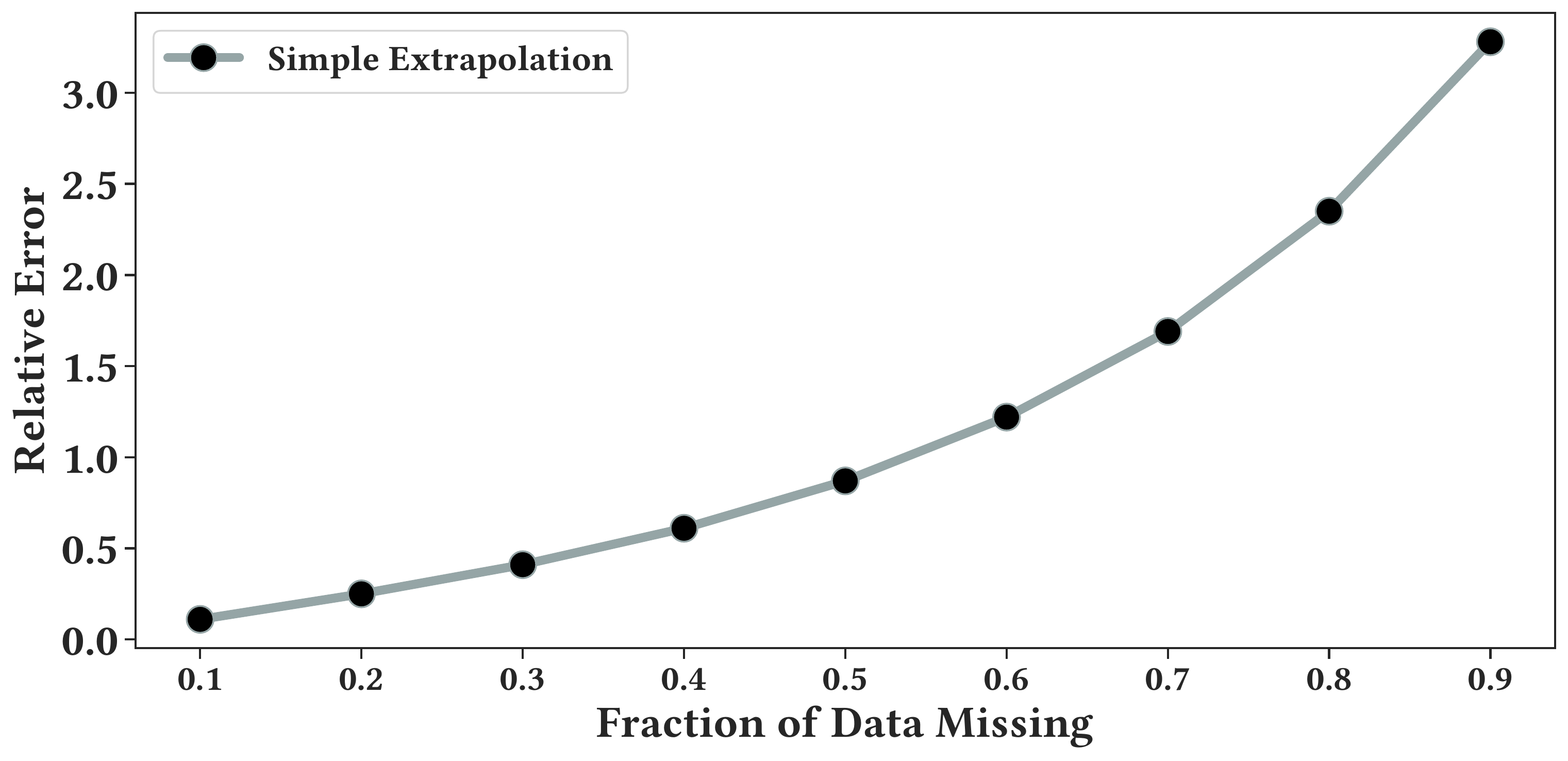}
    \caption{Simple Extrapolation could introduce significant error when missing rows are correlated (e.g, tend to have the highest values).}
    \label{se-teaser}
\end{figure}

\vspace{0.25em} \noindent \onethree{\textbf{Better Extrapolation: } A smarter approach might be to build a probabilistic model that identifies trends and correlations in the data (e.g., a Gaussian Mixture Model that identifies weekly patterns) and use that model to extrapolate. 
If a user mis-specifies her belief in the data distribution or sampling process, any inference would be equally fallible as simple extrapolation. The probabilistic nature of the inference also makes potential failure models hard to interpret---errors could arise due to modeling, sampling, or even approximation error in the model fitting process.}

\vspace{0.25em} \noindent \onethree{\textbf{Our Approach: } In light of these issues, we propose a fully deterministic model for quantifying the uncertainty in a query result due to missing data. Like the probabilistic approach, we require that the user specify her belief about the distribution of missing data. Rather than specifying these beliefs in terms of probability distributions, she specifies the beliefs in terms of hard constraints. \\
For example, there are no more than 300 sales each day in Chicago. Or, the most expensive product costs $149.99$ and no more than 5 are sold each day. We collect a system of such constraints, and solve an optimization problem to find the maximal sum possible for all missing data instances that satisfy the constraints. This formalism acts as a programming framework that the analyst can use to test the effects of different scenarios. It crucially enforces that there are testable constraints that are recorded during the decision making process. We will use this as an example throughout the paper.}

\section{Predicate-Constraints}
\label{sec:pc}
\shepherding{The formal problem setting is defined as follows.
Let $R$ be relation with a known ``certain'' partition denoted by $R^*$ and unknown ``missing'' partition $R^?$, where $R = R^* \cup R^?$.
The user defines a set of constraints $\pi_1,...,\pi_n$ over the possible tuples that could be in $R^?$ and their multiplicity.
The computational problem is to derive a \ci, namely, the min and max value that one of the supported aggregate queries given all possible instances of $R^?$ that are valid under the constraints.
This section describes our language for expressing constraints.
Suppose, this relation is over the attributes $A=\{a_1,...,a_p\}$. The domain of each attribute $a_i$ is a set denoted by $\textsf{dom}(a_i)$.}

\subsection{Predicate-Constraint}
\label{sec:pc-exp}
\shepherding{If $R^?$ could be arbitrary, there is clearly no way to bound its effect on an aggregate query. Predicate-constraints restrict the values of tuples contained in $R^?$ and the total cardinality of such tuples.
A single predicate-constraint defines a condition over the $R^?$, for tuples that satisfy the condition a range of attribute values, and the minimum and maximum occurrence of this predicate.
As an example predicate-constraint in the sales dataset described in the previous section, ``the most expensive product in Chicago costs 149.99 and no more than 5 are sold''. This statement has a predicate (``in Chicago''), a constraint over the values (``cost at most 149.99''), and occurrence constraint (``no more than 5''). We will show how to express systems of such constraints in our framework.}

\vspace{0.5em} \noindent \textbf{Predicate: } A predicate $\psi$ is a Boolean function that maps each possible rows to a True and False value $\psi:\mathcal{D} \mapsto \mathbb{Z}_2 $. For efficient implementation, we focus on predicates that are conjunctions of ranges and inequalities. This restriction simplifies satisfiability testing, which is an important step in our algorithms \reviewthree{introduced in Section \ref{sec:cell-decomp}}.

\vspace{0.5em} \noindent \textbf{Value Constraint: } A value constraint specifies a set of ranges that each attribute can take on. A range of the attribute $a_i$ is defined as a row of two elements $(l,h) \in \textsf{dom}(a_i)$ where $l \leq h$. A value constraint $\nu$ is a set of ranges for each of the $p$ attributes:
\[
\nu = \{(l_1,h_1),...,(l_p,h_p)\}
\]
$\nu$ defines a Boolean function as above $\nu:\mathcal{D} \mapsto \mathbb{Z}_2 $ that checks whether a row satisfies all the specified ranges. 
Since we focus on bounding aggregates it is sufficient to assume that the attribute ranges are over numerical attributes.


\vspace{0.5em} \noindent \textbf{Frequency Constraint: } Associated with each predicate is additionally a frequency constraint. This bounds the number of times that rows with the predicate appear. The frequency constraint is denoted as $\kappa = (k_l,k_u)$. \shepherding{$k_l$ and $k_u$ specify a range of the frequency constraint, i.e., there are at least $k_l$ rows and at most $k_u$ rows that satisfy the predicates in $R^?$. Of course, $k_l, k_u$ must be non-negative numbers and $k_l \leq k_u$.}


\vspace{0.5em} \noindent \textbf{Predicate Constraint: } A predicate-constraint is a three-tuple of these constituent pieces $\pi = (\psi, \nu, \kappa)$, a predicate, a set of value constraints, and a frequency constraint. The goal of $\pi$ is to define constraints on relations that satisfy the above schema. 
Essentially a predicate constraint says if $R$ is a relational instance that satisfies the predicate-constraint $\pi$ \textbf{``For all rows that satisfy the predicate $\psi$, the values are bounded by $\nu$ and the number of such rows is bounded by $\kappa$''}. Formally, we get the definition below.

\begin{definition}[Predicate Constraint]
A predicate constraint is a three-tuple consisting of a predicate, a value constraint, and a frequency constraint $\pi = (\psi, \nu, \kappa)$.
Let $R$ be a relational instance over the attributes $A$. $R$ satisfies a predicate constraint denoted by $R \models \pi$ if:
\reviewthree{
\[
(\forall r \in R: \psi(r) \implies \nu(r)) ~~~ \wedge 
k_l \le |\{r \in R: \psi(r)\}| \le k_u
\]
} 
\end{definition}
\reviewthree{Let us now consider several examples of predicate constraints using the sales data example in the previous section on two days worth of missing rows. For example, we want to express the constraint ``the most expensive product in Chicago costs $149.99$ and no more than 5 are sold'':
\[
c_1: \text{(branch = `Chicago')} \implies \text{(0.00} \leq \text{price} \leq \text{149.99)} ,~ (0, 5) 
\]
\shepherding{In this example, the predicate is $\text{(branch = `Chicago')}$, the value constraint is $\text{(0.00} \leq \text{price} \leq \text{149.99)}$ and the frequency constraint is (0, 5). If the aforementioned predicate constraint is describing a sales dataset, then it specifies that there are at most 5 tuples in the dataset with value of the $branch$ attribute equals to `Chicago' and the range of the $price$ attribute of these tuples are between 0.0 and 149.99 (inclusive).}

We could tweak this constraint to be ``the most expensive product in ALL branches costs $149.99$ and no more than 100 are sold'':
\[
c_2: \text{TRUE} \implies \text{(0.00} \leq \text{price} \leq \text{149.99)} ,~ (0, 100) 
\]
Suppose one wanted to define a simple histogram over a single attribute based the number of sales in each branch, we could express that with a tautology:
\[
\text{(branch = `Chicago')} \implies \text{(branch = `Chicago')} ,~ (100, 100)
\]
\[
\text{(branch = `New York')} \implies \text{(branch = `New York')} ,~ (20, 20)
\]
\[
\text{(branch = `Trenton')} \implies \text{(branch = `Trenton')} ,~ (10, 10)
\]
Observe how $c1$ and $c2$ interact with each other. Some of the missing data instances allowed by only $c2$ are disallowed by $c_1$ (Chicago cannot have more than 5 sales at 149.99). 
This interaction will be the main source of difficulty in computing \cis based on a set of PCs.}

\subsection{Predicate-Constraint Sets}\label{closure}
Users specify their assumptions about missing data using a set of predicate constraints
A predicate-constraint set is defined as follows:
\[
        \mathbf{S} = \{\pi_1,...,\pi_n\}
\]
$\mathbf{S}$ gives us enough information to bound the results of common aggregate queries when there is \emph{closure}: every possible missing row satisfies at least one of the predicates.
        \begin{definition}[Closure]
        Let $\mathbf{S}$ be a predicate constraint set with the elements $\pi_i = (\psi_i,\nu_i, \kappa_i)$
        $\mathbf{S}$ is closed over an attribute domain $\mathcal{D}$ if for every $t \in \mathcal{D}$:
        \[
        \exists \pi_i \in \mathbf{S}: \psi_i(t)
        \]
        \end{definition}
Closure is akin to the traditional closed world assumption in logical systems, namely, the predicate constraints completely characterize the behavior of the missing rows over the domain. 

\reviewtwo{To understand closure, let us add a new constraint that says ``the most expensive product in ``New York'' is 100.00 and no more than 10 of them are sold:
\[
c_3: \text{(branch = `New York')} \implies \text{(0.00} \leq \text{price} \leq \text{100.00)} ,~ (0, 10) 
\]
Closure over ${c_1,c_3}$ means that all of the missing rows come from either New York or Chicago.}

\section{Calculating \cis}
This section focuses on a simplified version of the bounding problem. We consider a single table and a single attribute aggregates. Let $q$ denote such an aggregate query. The problem to calculate the upper bound is:
\begin{equation}
u = \max_{R} q(R)
\end{equation}
\[
\text{subject to: } R \models \mathbf{S}
\]
We will show that our bounds are tight--meaning that the bound found by the optimization problem is a valid relation that satisfies the constraints.

Throughout the rest of the paper, we only consider the maximal problem. Unless otherwise noted, our approach also solves the lower bound problem:
\[
l = \min_{R \in \mathcal{R}} q(R)
\]
\[
\text{subject to: } R \models \mathbf{S}
\] 
Specifically, we solve the lower bound problem in two settings. In a general setting where there is no additional constraint, the minimal problem is can be solved by maximizing the negated problem: we first negate the value constraints, solve the maximizing problem with negated weights, and negate the final result.

In a special but also common setting, all the frequency constraints' lower bounds are $0$ (i.e., each relation has no minimum number of missing rows) and the value constraints' lower bounds are $0$ (i.e., all attributes are non-negative), the lower bound is trivially attained by the absent of missing row.

\begin{figure}
    \centering
    \includegraphics[width=\columnwidth]{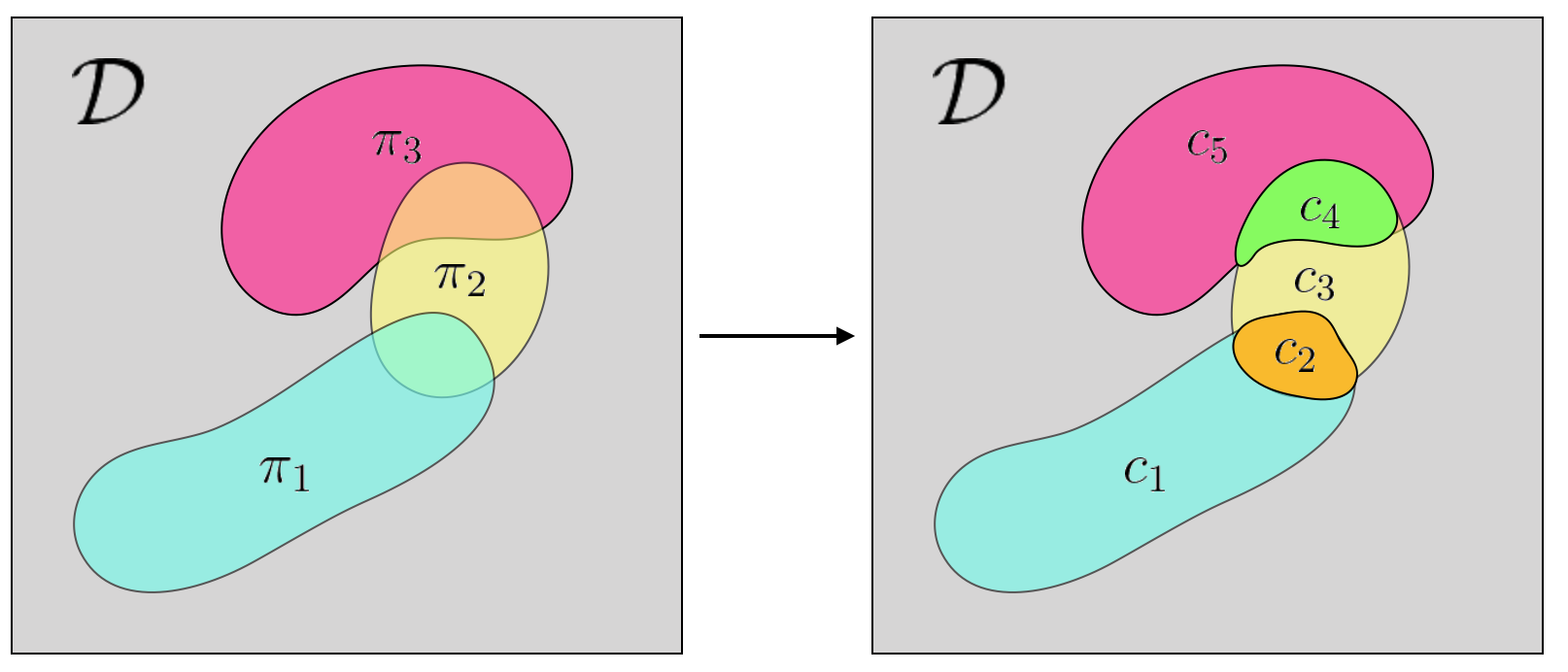}
    \caption{Predicates in a predicate-constraint set are possibly overlapping. The first step of the algorithm is to decompose a set of predicate-constraints into disjoint cells.}
    \label{fig:cell-decomp}
\end{figure}

\subsection{Cell Decomposition}\label{sec:cell-decomp}
Intuitively, we can think about the optimization problem as an allocation. We strategically assign rows in each of the predicates to maximize the aggregate query. However, the first challenge is that a row may fall in multiple Predicate-Constraints' predicates, so this row may ``count towards'' multiple Predicate-Constraints. As the first step of the solution, we decompose the potential overlapping Predicate-Constraints' predicates into disjoint cells.

An example that illustrates the decomposition is depicted in Figure \ref{fig:cell-decomp}. Each predicate represents a sub-domain. For each subset of predicates, a cell is a domain that only belongs to these predicates and not others. Thus, for $n$ predicates in a predicate-constraint set there are at most $O(2^n)$ cells. The cells take the form of conjunctions and negations of the predicates of each of the $n$ predicate constraints:
\[
c_0 = \psi_1 \wedge ... \wedge \psi_n
\]
\[
c_1 = \psi_1 \wedge ... \wedge \neg \psi_n  
\]
\[
c_2 = \psi_1 \wedge ...\wedge \neg \psi_{n-1} \wedge \psi_n 
\]
\[
...
\]
\[
c_{2^n-1} = \neg \psi_1 \wedge ...\wedge \neg \psi_{n-1} \wedge \neg\psi_n  
\]
For each $c_i$, we have to reconcile the active predicate constraints (not negated above).
Each cell is thus assigned the most restrictive upper and lower value bounds, and upper and lower cardinality bounds in the set of active constraints.
Not all possible cells will be satisfiable--where there exists a row $t \in \mathcal{D}$ that satisfies the new predicate-constraint. As in Figure \ref{fig:cell-decomp}, there are $7$ possible subsets, but there are only $5$ satisfiable cells. We use the Z3~\cite{Z3} solver to prune all the cells that are not satisfiable.

\stitle{Optimization 1. Predicate Pushdown:}  Cell decomposition is a costly process for two reasons. First, there is a potentially exponential number of cells. Second, determining whether each cell is satisfiable is not always easy (each check is on the order of 10's of ms). 
One obvious optimization is to push down the predicates of the target query into the decomposition process.
When the target query has predicates, we exclude all cells that do not overlap with the query's predicate.

\stitle{Optimization 2. DFS Pruning:}
The naive solution is to simply sequentially iterate through all of the cells and test each for satisfiability.
Note that for a problem with $n$ PCs, each logical expression describing a cell is a conjunction of $n$ predicates.
Conjunctions can be short-circuited if any of their constituent elements evaluate to false. 
Therefore, the process of evaluating the satisfiability of the cells can be improved by using a Depth First Search that evaluates prefixes of growing length rather than a sequential evaluation. 
 With DFS, we can start from the root node of all expressions of length 1, add new PCs to the expression as we traverse deeper until we reach the leaf nodes which represent expressions of length $n$. As we traverse the tree, if a sub-expression is verified by Z3 to be unsatisfiable, then we can perform pruning because any expression contains the sub-expression is unsatisfiable. 

\stitle{Optimization 3. Expression Re-writing:}  To further improve the DFS process, we can re-write the logical expressions to have an even higher pruning rate. There is one simple re-writing heuristic we apply:
\[
(X \wedge \neg(X\wedge Y)) = True \implies X \wedge \neg Y = True
\]
It means if we have verified a sub-expression $X$ to be satisfiable, and we also verified that after adding a new PC $Y$, the expression $X\wedge Y$ becomes unsatisfiable. We can conclude that $X\wedge \neg Y$ is satisfiable without calling Z3 to verify.
As shown in the experiment section, the DFS pruning technique combined with the rewriting can prune over 99.9\% cells in real-world problems.

\stitle{Optimization 4. Approximate Early Stopping:} 
With the DFS pruning technique, we always get the correct cell decomposition result because we only prune cells that are verified as unsatisfiable.  We also propose an approximation that can trade range tightness for a decreased run time. The idea is to introduce `False-Positives', i.e., after we have used DFS to handle the first $K$ layers (sub-expressions of size $K$), we stop the verification and consider all cells that have not been pruned as satisfiable. These cells are then admitted to the next phase of processing where we use them to formulate a \milp problem that can be solved to get our bound. Admitting unsatisfiable cells introduces `False-Positives' that would make our bound loose, but it will not violate the correctness of the result (i.e. the result is still a bound) because: (1) the `true-problem' with correctly verified cells is now a sub-problem of the approximation and (2) the false-positive cells does not add new constraints to the `true-problem'.

\subsection{Integer-Linear Program}\label{sec:ilp}


We assume that the cells in $\mathcal{C}$ are ordered in some way and indexed by $i$. Based on the cell decomposition, we denote $C_i$ as the (sub-)set of Predicate-Constraints that cover the cell $i$. Then for each cell $i$, we can define its maximal feasible value $U_i(a)=\min_{p\in{}C_i} p.\nu.h_a$, i.e., the minimum of all $C_i$'s value constraints' upper bounds on attribute $a$. 

A single general optimization program can be used to bound all the aggregates of interest.
Suppose we are interested in the \texttt{SUM} over attribute $a$, we slightly abuse the terms and define a vector $U$ where $U_i=U_i(a)$. Then, we can define another vector $X$ which represents the decision variable. Each component is an integer that represents how many rows are allocated to the cell.

The optimization problem is as follows. To calculate the upper bound:
\begin{equation}
    \max_{X} ~~ U^T X
    \label{prob:opt}
\end{equation}
\[
\text{subject to: } ~ 
 ~ \forall j,  k^{(j)}_l \le \sum_{i : j \in C_i} X[i] \le k^{(j)}_u\\
\]
\[
\forall i ~ , ~ X[i] \text{ is integer}
\]
As an example, the first constraint in Figure \ref{fig:cell-decomp} is $k_l^1 <= x_1 + x_2 <= k_u^1$ for $\pi_1$, so on and so forth. 

Given the output of the above optimization problem, we can get bounds on the following aggregates:

\vspace{0.5em} \noindent \textbf{COUNT: } The count of cardinality can be calculated by setting \reviewthree{$U$ as the unit vector ($U_i=1$ for all $i$)}.

\vspace{0.5em} \noindent \textbf{AVG: } We binary search the average result: to testify whether $r$ is a feasible average, we disallow all rows to take values smaller than $r$ and invoke the above solution for the maximum \texttt{SUM} and the corresponding \texttt{COUNT}. If the overall average is above $r$, then we test $r'>r$, otherwise we continue testing $r'<r$. 


\vspace{0.5em} \noindent \textbf{MAX/MIN: } Assuming all cells are feasible, the max is the largest of all cells' upper bound. Min can be handled in a similar way.

\stitle{Faster Algorithm in Special Cases}  In the case that the Predicate-Constraints have disjoint predicates, the cell decomposition problem becomes trivial as each predicate is a cell. The MILP problem also degenerates since all the constraints, besides each variable being integer,  do not exist at all. Thus, if we take the \texttt{SUM} problem as an example, the solution is simply the sum of each Predicate-Constraints' maximum sum, which is the product of the maximum value and the maximum cardinality. This disjoint case is related to the work in data privacy by Zhang et al.~\citep{zhang2007aggregate}.

\subsection{Complexity} \label{sec:complexity}
\shepherding{Next, we analyze the complexity of the predicate-constraint optimization problem. Suppose each predicate-constraint is expressed in Disjunctive Normal Form with at most $p$ clauses. Suppose, we are given a set of $N$ such predicate-constraints. The data complexity of the predicate-constraint optimization problem is the number of computational steps required to solve the optimization problem for a fixed $p$ and a variable $N$. For $p >=2$, we show that the problem is computationally hard in $N$ based on a reduction:}

\begin{proposition}
\label{proposition41}
Determining the maximal sum of a relation constrained by a predicate-constraint is NP-Hard.
\end{proposition}

\begin{proof}[Sketch] We prove this proposition by reduction to the maximal independent set problem, which is NP-Hard. An independent set is a set of vertices in a graph, no two of which are adjacent. We can show that every independent set problem can be described as a predicate-constraint maximal sum problem. Let $G = (V,E)$ be a graph. We can think of this graph as a relational table of ``vertex'' rows. For example: (V1, V2). For each vertex $v \in V$, we define a single Predicate-Constraint with a max value of 1 and a max frequency of 1 $(x = v, [0,1], [0,1])$. For each edge between $v$ and $v'$, we define another predicate constraint that is equality to either of the vertices also with a max frequency of 1: \[(x = v \vee x = v', [0,1], [0,1]).\] This predicate constraint exactly contains the two vertex constraints. Therefore, the cells after the cell-decomposition step perfectly align with those vertices. The optimization problem allocates rows to each of those cells but since the edge constraint has a max frequency is 1 only one of the vertex cells can get an allocation. Since all of the vertex constraints have the same max value, the optimization problem finds the most number of such allocations that are consistent along the edges (no two neighboring vertices both are allocated). This is exactly the definition of a maximal independent set. Since the maximal sum problem is more expressive than the maximal independent set its complexity is greater than NP-Hard in the number of cells.\end{proof}

\shepherding{The hardness of this problem comes from the number of predicate constraints (i.e., data complexity). While, we could analyze the problem for a fixed $N$ and variable $p$ (i.e., query complexity), we find that this result would not be informative and is highly problem specific with very large constant-factors dependent on $N$. }

\subsection{Numerical Example}\label{quant-ex}
\onethree{To understand how this optimization problem works, let's consider a few simple numerical examples using our example dataset. Suppose, we have a query that wants to calculate the total number of sales over a period: }
\begin{lstlisting}
SELECT SUM(price)
FROM Order
WHERE utc >= Nov-11 0:00 AND
      utc <= Nov-13 0:00
\end{lstlisting}
\onethree{For the sake of simplicity, let's assume that all of this data was missing and the only information we have is what is described in the PCs.
Suppose, we define two PCs that describe the price of the least/most expensive items sold on the day, and that between 50 and 100 items were sold:
\[\scriptstyle
t_1: Nov-11 <= utc < Nov-12 \implies 0.99 <= \text{price} <= 129.99  ,~ (50, 100)
\]
\[\scriptstyle
t_2: Nov-12 <= utc < Nov-13 \implies 0.99 <= \text{price} <= 149.99  ,~ (50, 100)
\]
Since the PCs are disjoint, we can trivially compute the result range for the total sales:
\[
[50\times0.99+50\times0.99, 100\times129.99 + 100\times149.99]
\]
\[
 = [99.00, 27998.00] \blacksquare
\]
}
\onethree{\noindent Now, suppose, we had overlapping PCs. This makes the solution much harder to manually reason about since we have to account for the interaction between the constraints. 
\[\scriptstyle
t_1: Nov-11 <= utc < Nov-12 \implies 0.99 <= \text{price} <= 129.99  ,~ (50, 100)
\]
\[\scriptstyle
t_2: Nov-11 <= utc < Nov-13 \implies 0.99 <= \text{price} <= 149.99  ,~ (75, 125)
\]
This requires a cell decomposition to solve. There are 3 possible cells:
\[\scriptstyle
c_1 := (Nov-11 <= utc < Nov-12) \wedge (Nov-11 <= utc < Nov-13)
\]
\[\scriptstyle
c_2 := \neg (Nov-11 <= utc < Nov-12) \wedge (Nov-11 <= utc < Nov-13)
\]
\[\scriptstyle
c_3 :=  (Nov-11 <= utc < Nov-12) \wedge \neg (Nov-11 <= utc < Nov-13)
\]
$c_3$ is clearly not satisfiable, so we can discount this cell. 
Then, we need to figure out how much to allocate to $c_1$ and $c_2$.
The lower bound can be achieved by an allocation of $50$ tuples to $c_1$ and $25$ to $c_2$.
The upper bound can be achieved by an allocation of $50$ tuples to $c_1$ and $75$ tuples to $c_2$. Note that the optimal allocation does not maximize tuples in $c_1$.
\[
[50\times0.99+25\times0.99, 50\times129.99 + 75\times149.99] =
\]
\[
[74.25, 17748.75] \blacksquare
\]
}

\section{Joins Over Predicate Constraints}
\label{sec:joins}
In the previous section, we considered converting single table predicate-constraints into \cis for a query result. In this section, we extend this model to consider aggregate queries with inner join conditions, namely, there are predicate constraints describing missing data in each of the base tables and we have to understand how these constraints combine across tables.

\subsection{Naive Method}
One way to handle multi-relation predicate constraints is to treat a join as a Cartesian product. 
Let's consider two tables $R$ and $S$, and let $\mathcal{P}_R$ and $\mathcal{P}_S$ denote their predicate constraints sets respectively.
For two predicate constraints from each set $\pi_s$ and $\pi_r$, let's define a direct-product operation as:
\[
\pi_s \times \pi_r = (\psi_s \wedge \psi_r,[\nu_s \nu_r], \kappa_s \otimes \kappa_r)
\]
that takes a conjunction of their predicates, concatenates their attribute ranges, and multiplies their cardinalities.
If we do this for all pairs of predicate constraints, we can derive a set of constraints that account for the join:
\[\mathcal{P} = \{\pi_i \times \pi_j : \forall \pi_i, \pi_j \in \mathcal{P}_R, \mathcal{P}_T \}\]
This approach will produce a bound for all inner-joins since any satisfying tuple in the output has to be satisfying either $\mathcal{P}_R$ or $\mathcal{P}_T$.

While this approach will produce a bound, it may be very loose in certain cases.
It is particularly loose in the case of inner equality joins.
Most obviously, it does not consider the effects of equality conditions and how those may affect cardinalities and ranges.
For conditions that span more than 2 relations, we run into another interesting problem.
The Worst-Case Optimal Join (WCOJ) results are some of the most important results in modern database theory~\cite[Lemma 3.3]{WCOJ}. Informally, they show that solving an n-way join with a cascade of two-way joins can create exponentially more work due to very large intermediate results.
An optimal algorithm would only do work proportional to the number of output rows and no more.

A very similar issue arises with this bounding problem.
Consider the exemplary triangle counting query:
 \[q = |R(a,b)S(b,c)T(c,a)|\]
Suppose, each relation has a size of $N$.
If we apply the naive technique to bound $q$ for a predicate-constraint set defined for each relation, the bound would be $O(N^{3})$.
However, from WCOJ results, we know that the maximum value of q is $O(N^{\frac{3}{2}})$~\cite{ngo2018worst}. 
We can perpetuate this logic to the 4-clique counting query, 5-clique, and so on:
 \[q = |R(a,b,c)S(b,c,d)T(c,d,e)U(e,a,b)|\] and the gap between the theoretical bound and the one computed by our framework grows exponentially.

\subsection{A Better Bound For Natural Joins}
We can leverage some of the theoretical machinery used to analyze WCOJ algorithms to produce a tighter bound. We first introduce an important result for this problem, Friedgut's  Generalized Weighted Entropy (GWE) inequality~\cite[Lemma 3.3]{WCOJ}.

There are $r$ relations, indexed by $i$ as $R_i$. The joined relation is $R$. For each row $t$ in $R$, its projection in the relation $R_i$ is $t^i$. For each relation $R_i$, this is an arbitrary non-negative weight function $w_i(\cdot)$ defined on all rows from $R_i$. Fractional edge cover (FEC) is a vector $c$ that assigns a non-negative value $c_i$ to each relation $R_i$, and also satisfies that for each attribute $s$, \[\sum_{R_i\oplus{}s}c_i\geq1\] $R_i\oplus{}s$ if the relation $R_i$ contains attribute $s$ (when multiple relations join on one attribute, we consider the attribute indistinguishable, i.e., they all contain the same attribute).

When $c$ is FEC, GWE states
\[\sum_{t\in{}R}\prod_{i}\big(w_i(t^i)\big)^{c_i} \leq \prod_i\big(\sum_{t_i\in{}R_i}w_i(t_i)\big)^{c_i} \tag{*}\]
What GWE implies depends on the choice of weight functions $w_i(\cdot)$. In a query with \texttt{SUM(A)} aggregation, without loss of generality we assume attribute $A$ comes from the relation $R_a$. Let $w_a(t_a)=t_a.A$, and $w_i(t_i)=1$ for $i\neq{}a$. Then the left hand of $(*)$ becomes \[\sum_{t\in{}R}w_a(t^a)^{c_a}=\sum_{t\in{}R}t.A^{c_a}\] and the right hand of $(*)$ becomes
\[\prod_{i\neq{}a}|R_i|^{c_i}\times\big(\sum_{t_a\in{}R_a}t_a.A\big)^{c_a}\] We only consider the FEC $c$ such that $c_a=1$, then (*) becomes 
\[\sum_{t\in{}R}t.A \leq \prod_{i\neq{}a}|R_i|^{c_i}\times\big(\sum_{t_a\in{}R_a}t_a.A\big) \tag{**}\]

Now given a set of PC $\pi_i$ for each relation $R_i$, and we only consider those $R_i$ that conform to corresponding PCs. According to $(**)$, we have 
\[\underset{\substack{t_1,\dots,t_r\text{~natural join}\\ R_i \models \pi_i}}{\text{\texttt{SUM}}(A)} \leq \underset{ R_a\models{}\pi_a}{\text{\texttt{SUM}}(A)} \times \prod_{i\neq{}a}\big(\underset{ R_i\models{}\pi_i}{\text{\texttt{COUNT}}(*)}\big)^{c_i}   \]
Here the left hand side is our expected result, and the right hand side can be solved on each relation individually using the approaches discussed in the previous section.

Note that we use the solution to the FEC problem to compute $c$, and this solution derives an upper bound of the left-hand side. In order to get the tightest upper bound, we consider an optimization problem: minimize the right-hand side subject to that $c$ is an FEC. We take a $\log$ of the target function (i.e., right-hand side) so both the target function and constraints are in linear form. The optimization problem becomes a linear programming problem, which can be solved by a standard linear programming solver. See Section 6.4 for two cases where this solution significantly improves on the naive Cartesian product solution described above.

\section{Results}
Now, we evaluate the PC framework in terms of accuracy and efficiency at modeling missing data rows.

\subsection{Experimental Setup}\label{exp-setup}
We set up each experiment as follows: (1) summarize a dataset with \reviewthree{$n$} Predicate-Constraints, (2) each competitor framework gets a similar amount of information about the dataset (e.g., a statistical model with $O(n)$ statistical parameters), (3) we estimate a query result using each framework with calibrated error bounds. 
Our comparative evaluation focuses on SUM, COUNT queries as those are what the baselines support.
We show similar results on MIN, MAX, AVG queries, but only within our framework. 

\reviewone{It is important to note that each of these frameworks will return a confidence interval and not a single result. Thus, we have to measure two quantities: 
the \textbf{failure rate} (how often the true result is outside the interval) and the tightness of the estimated ranges. 
We measure tightness by the ratio between the upper bound and the result (we denote this as the \textbf{over estimation rate}). A ratio closer to $1$ is better on this metric, but is only meaningful if the failure rate is low.}

\subsubsection{Sampling}\label{sampling}\reviewtwo{
To construct a sampling baseline, we assume that the user provides actual unbiased example missing data records.
While this might be arguably much more difficult for a user than simply describing the attribute ranges as in a predicate constraint, we still use this baseline as a measure of accuracy. In our estimates, we use \emph{only} these examples to extrapolate a range of values that the missing rows could take.} 

\vspace{0.25em} \noindent \textbf{Uniform Sampling} We randomly draw $n$ samples (US-1) and 10$n$ samples (US-10) from the set of missing rows.

\vspace{0.25em} \noindent \textbf{Stratified Sampling} We also use a stratified sampling method which performs a weighted sampling from partitions defined by the PCs that we use for a given problem. Similarly, we denote $n$ samples as (ST-1) and 10$n$ samples (ST-10).

\reviewtwo{The goal is to estimate the result range using a statistical confidence interval. Commonly, approximate query processing uses the Central Limit Theorem to derive bounds.  These confidence intervals are parametric as they assume that the error in estimation is Normally distributed. Alternatively, one could use a non-parametric method, which does not make the Normal assumption, to estimate confidence intervals like those described in~\cite{hellerstein1997online} (we use this formula unless otherwise noted). We denote confidence interval schemes in the following way: US-1p (1x sample using a parametric confidence interval), US-10n (10x sample using a non-parametric confidence interval), ST-10p (10x stratified sample using a parametric confidence interval), etc.}

\subsubsection{Generative Model}
Another approach is to fit a generative model to the missing data. We use the whole dataset as training data for a Gaussian Mixture Model (GMM). The trained GMM is used to generate the missing data. A query result that is evaluated on the generated data is returned. This is simulating a scenario where there is a generative model that can describe what data is missing. If we run this process several times, we can determine a range of likely values that the query could take.

\subsubsection{Equiwidth Histogram} We build a histogram on all of the missing data on the aggregate attribute with $N$ buckets and use it to answer queries. We use standard independent assumptions to handle queries across multiple attributes.

\subsubsection{PCs}
The accuracy of the PC framework is dependent on the particular PCs that are used for a given task.
In our ``macro-benchmarks'', we try to rule out the effects of overly tuned PCs.
We consider two general schemes: \textbf{Corr-PC}, even partitions of attributes correlated with the aggregate of interest, and \textbf{Rand-PC}, randomly generated PCs.

For \textbf{Corr-PC}, we identify the (other than the aggregation attribute) most correlated attributes with the aggregate to partition on. 
We divide the combined space into equi-cardinality buckets where each partition contains roughly the same number of tuples.
We omit the details of this scheme for the sake of brevity.
For \textbf{Rand-PC}, we generate random overlapping predicate constraints over the same attributes (not necessarily equi-cardinality).
We take extra care to ensure they adequately cover the space to be able to answer the desired queries.
We see these as two extremes: the reasonably best performance one could expect out of the PC framework and the worst performance. 
We envision that natural use cases where a user is manually defining PCs will be somewhere in the middle in terms of performance.

\subsection{Intel Wireless Dataset}\label{sec:missing_ratio}
The Intel wireless dataset \cite{intelwireless} contains data collected from 54 sensors deployed in the Intel Berkeley Research lab in 2004. This dataset contains \reviewtwo{3 million rows}, 8 columns including different measurements like humidity, temperature, light, voltage as well as date and time of each record. We consider aggregation queries over the \texttt{light} attribute. For this dataset,  Corr-PC is defined as $n=2000$ predicate constraints over the attributes \texttt{device\_id} and \texttt{time}. Rand-PC defines random PCs over those same attributes.
Missing rows are generated from the dataset in a correlated way---removing those rows maximum values of the \texttt{light} attribute.

\reviewone{We compare Corr-PC and Rand-PC with 3 baselines we mentioned in Section \ref{exp-setup}: US-1n, ST-1n, and Histogram. We vary the fraction of missing rows $r$ and evaluate the accuracy and failure rate of each technique. Figure \ref{intel-mr-cnt} and Figure \ref{intel-mr-sum} illustrate the experimental results: (1) as per our formal guarantees, both Corr-PC, Rand-PC (and Histograms) do not fail if they have accurate constraints, (2) despite the hard guarantee, the confidence intervals are not too ``loose'' and are competitive or better than those produced by a 99.99\% interval, (3) informed PCs are an order of magnitude more accurate than randomly generated ones. There are a couple of trends that are worth noting. First, the failure rate for sampling techniques on the SUM queries is higher than the expected 1 in 10000 failures stipulated by the confidence intervals.
In this missing data setting, a small number of example rows fail to accurately capture the ``spread'' of a distribution (the extremal values), which govern failures in estimation.
Queries that require values like SUM and AVG are very sensitive to these extrema.}

\reviewone{It is important to note that these experiments are idealized in the sense that all the baselines get true information about the missing data and have to summarize this information into $O(n)$ space and measure how useful that stored information is for computing the minimal and maximal value a workload of aggregate queries could take. There is a subtle experimental point to note. If a query is fully covered by the missing data, we solve the query with each baseline and record the results; if a query is partially covered by the missing data, we solve the part that is missing with each baseline then combine the result with a `partial ground truth' that is derived from the existing data; finally, if a query is not overlapping with the missing data then we can get an accurate answer for such a query. Such issues are common to all the baselines, in the upcoming experiments we will only consider the accuracy of representing the missing data (and not partially missing data) to simplify questions of correlation and accuracy.}

\begin{figure}[ht]
    \centering
    \includegraphics[width=1\linewidth]{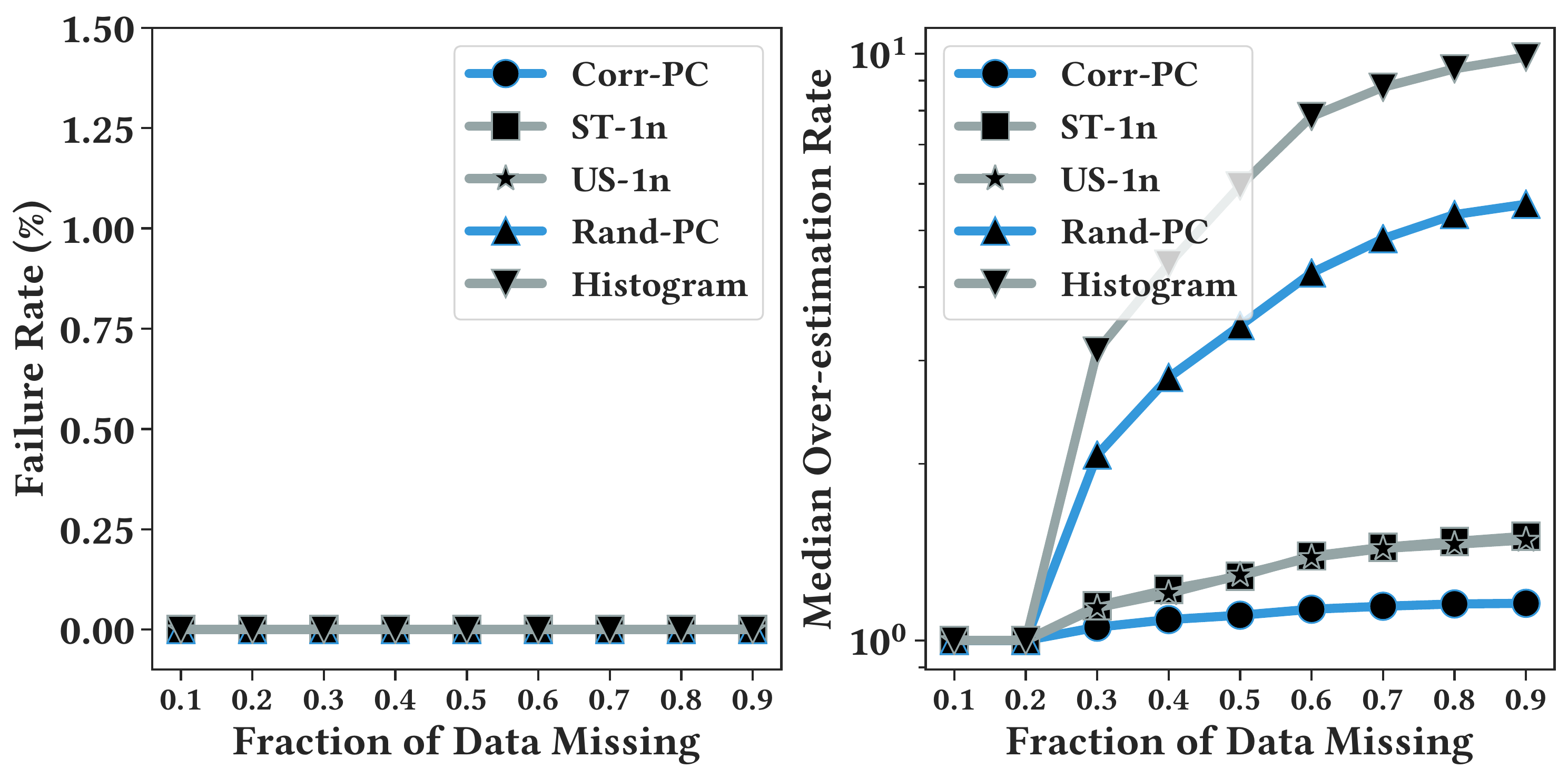}
    \caption{\changed{Performance of baselines given different missing ratio, evaluated with 1000 COUNT(*) query on the Intel Wireless dataset.}}
    \label{intel-mr-cnt}
\end{figure}

\begin{figure}[h]
    \centering
    \includegraphics[width=1\linewidth]{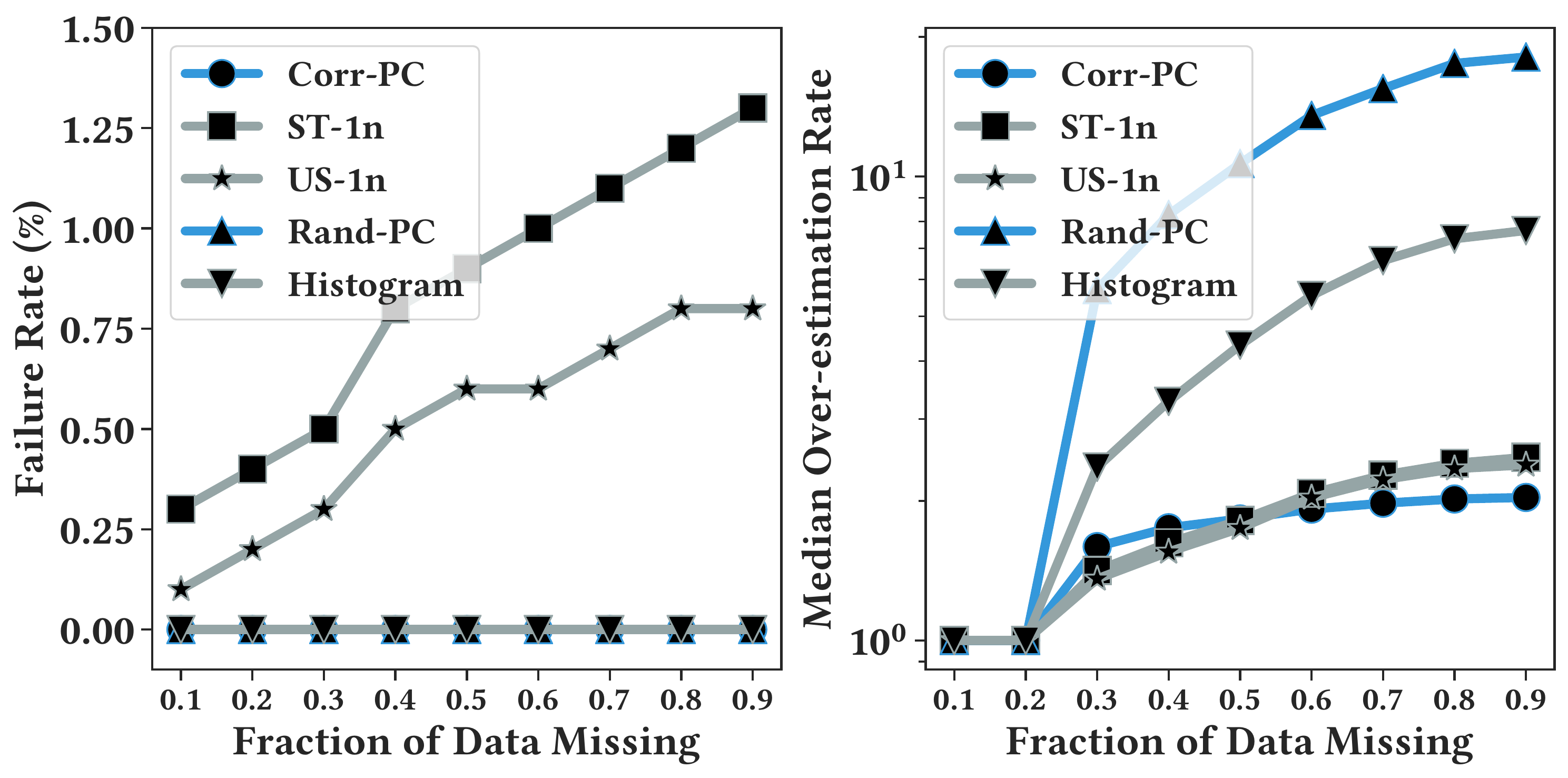}
    \caption{\changed{Performance of baselines given different missing ratio, evaluated with 1000 SUM query on the Intel Wireless dataset.}}
    \label{intel-mr-sum}
\end{figure}

\subsection{\changed{Detailed Sampling Comparison}}
\label{sec:fr-tradeoff}
\changed{One might argue that $99.99\%$ is an overly conservative confidence interval. In this experiment, we evaluate the performance of the uniform sampling baseline in terms of failure rate and accuracy as function of that confidence interval setting. Our results show there is not a clear way to calibrate the confidence intervals to either minimize failures or avoid inaccuracy. Results in Table \ref{clt-fr-median} show a clear trade-off between failure rate and accuracy: when we increase the confidence interval, the over-estimation ratio increases, and the failure rate decreases. However, even with a 99.99$\%$ confidence interval derived from Chernoff Bound, the failure rate is non-zero. We believe that the PC framework provides the user with a competitive accuracy but the guarantee of no failures.}

\begin{table}[ht]
    \begin{center}
    \begin{tabular}{llllllll}
    \hline
    \multicolumn{8}{c}{Failure Rate $\%$}\\
    \hline
    Conf (\%)        & 80   & 85   & 90   & 95   & 99  & 99.9 & 99.99 \\ \hline
    US-1n     & 20.1 & 15.6 & 11.4 & 6.9  & 3.4 & 2.4  & 0.8   \\
    Corr-PC & \multicolumn{7}{c}{--0--}        \\ 
    \hline
    \multicolumn{8}{c}{Over Estimation Rate}\\
    \hline
    US-1n     & 1.07 & 1.08 & 1.11 & 1.13 & 1.2 & 1.27 & 3.13  \\
    Corr-PC & \multicolumn{7}{c}{--2.23--}  \\ \hline
    \end{tabular}
    \end{center}
    \caption{\changed{Trade-off between failure rate and accuracy of an uniform sampling baseline given different confidence interval vs. Corr-PC.}}
    \label{clt-fr-median}
\end{table}

\subsubsection{Sampling More Data}\label{exp-data}
In all of our experiments above, we consider a 1x random sample. Where the baseline is given the same amount of data compared to the number of PCs.
PCs are clearly a more accurate estimate in the ``small data'' regime, what if the sample size was larger.
In Figure \ref{clt-chernoff}, we use the Intel Wireless dataset to demonstrate the performance of the non-parametric bounds using different sample sizes. And the results demonstrate a clear trend of convergence as we increase the sample size.
If we consider data parity (1x), the confidence interval is significantly less accurate than a well-designed PC.
One requires 10x the amount of data to cross over in terms of accuracy.
We envision that PCs will be designed by a data analyst by hand, and thus the key challenge is to evaluate the accuracy of the estimation with limited information about the distribution of missing values.

\begin{figure}[h]
    \centering
    \includegraphics[width=1\linewidth]{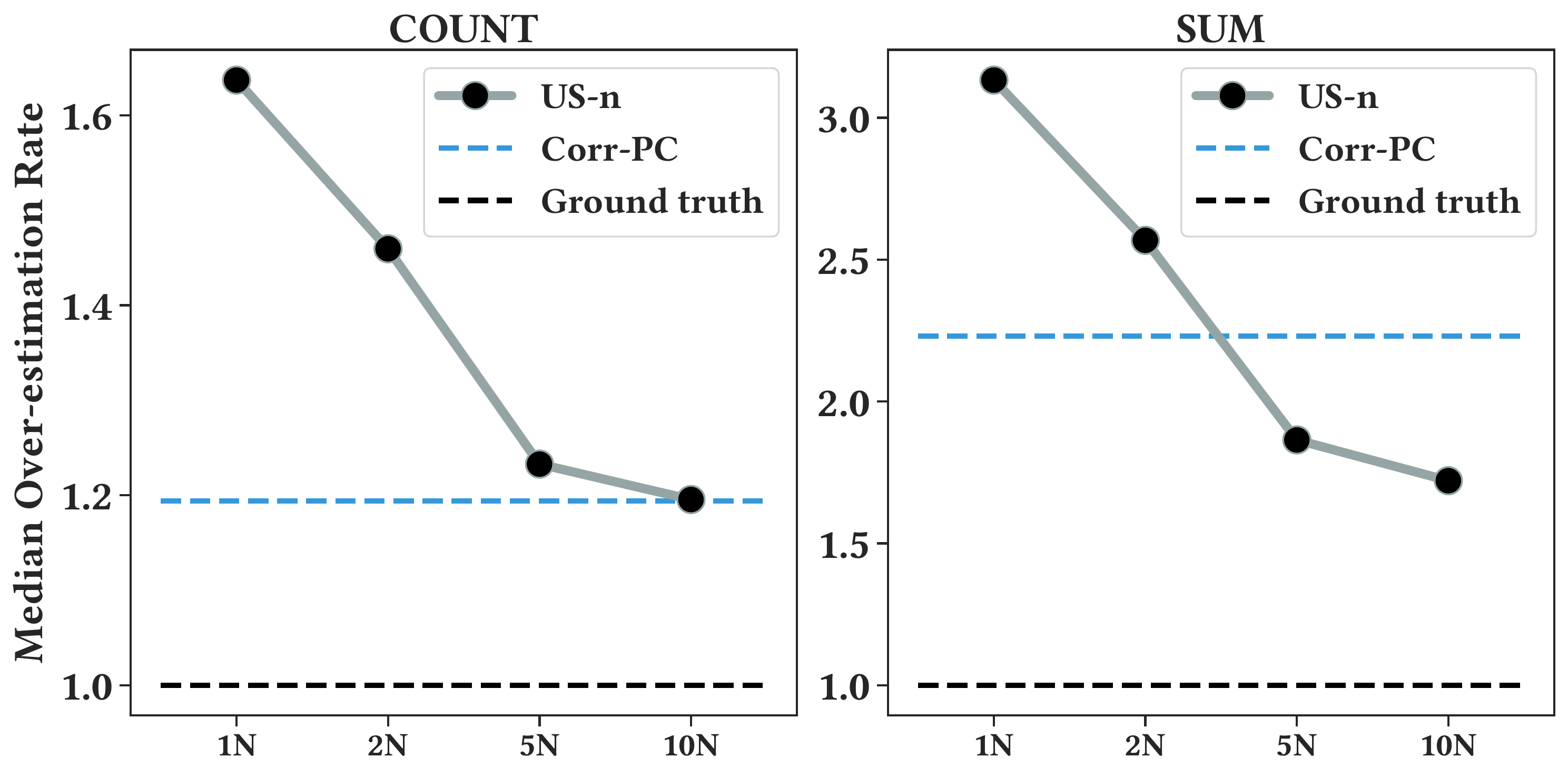}
    \caption{\changed{Performance of the uniform sampling baseline with different sample size.}}
    \label{clt-chernoff}
\end{figure}

\subsubsection{Robustness To Noise}\label{exp-robustness}
Of course, these results depend on receiving PCs that accurately model the missing data.
We introduce noise into the PCs to understand how incorrectly defined PCs affect the estimated ranges.
When the PCs are noisy, there are failures because the ranges could be incorrect. 
\reviewthree{We add independent noise to the minimum and maximum values for each attribute in each PCs.}
Figure \ref{noisy-partition} plots the results for Corr-PC, a set of 10 overlapping PCs (Overlapping-PC), and US-10n (a 10x sample from the previous experiment).
We corrupt the sampling bound by mis-estimating the spread of values (which is functionally equivalent to an inaccurate PC).
All experiments are on the SUM query for the Intel Wireless dataset as in our previous experiment.
For corrupting noise that is drawn from a Gaussian distribution of 1, 2 and 3 standard deviation, we plot the failure rate. 

Our results show that PCs are not any more sensitive than statistical baselines.
In fact, US-10n has the greatest increase in failures due to the noise. Corr-PC is significantly more robust. 
This experiment illustrates the benefits of overlapping PCs. 
When one such overlapping PC is incorrect, our framework automatically applies the most restrictive overlapping component.
This allows the framework to reject some amount of mis-specification errors.

\begin{figure}[h]
    \centering
    \includegraphics[width=1\linewidth]{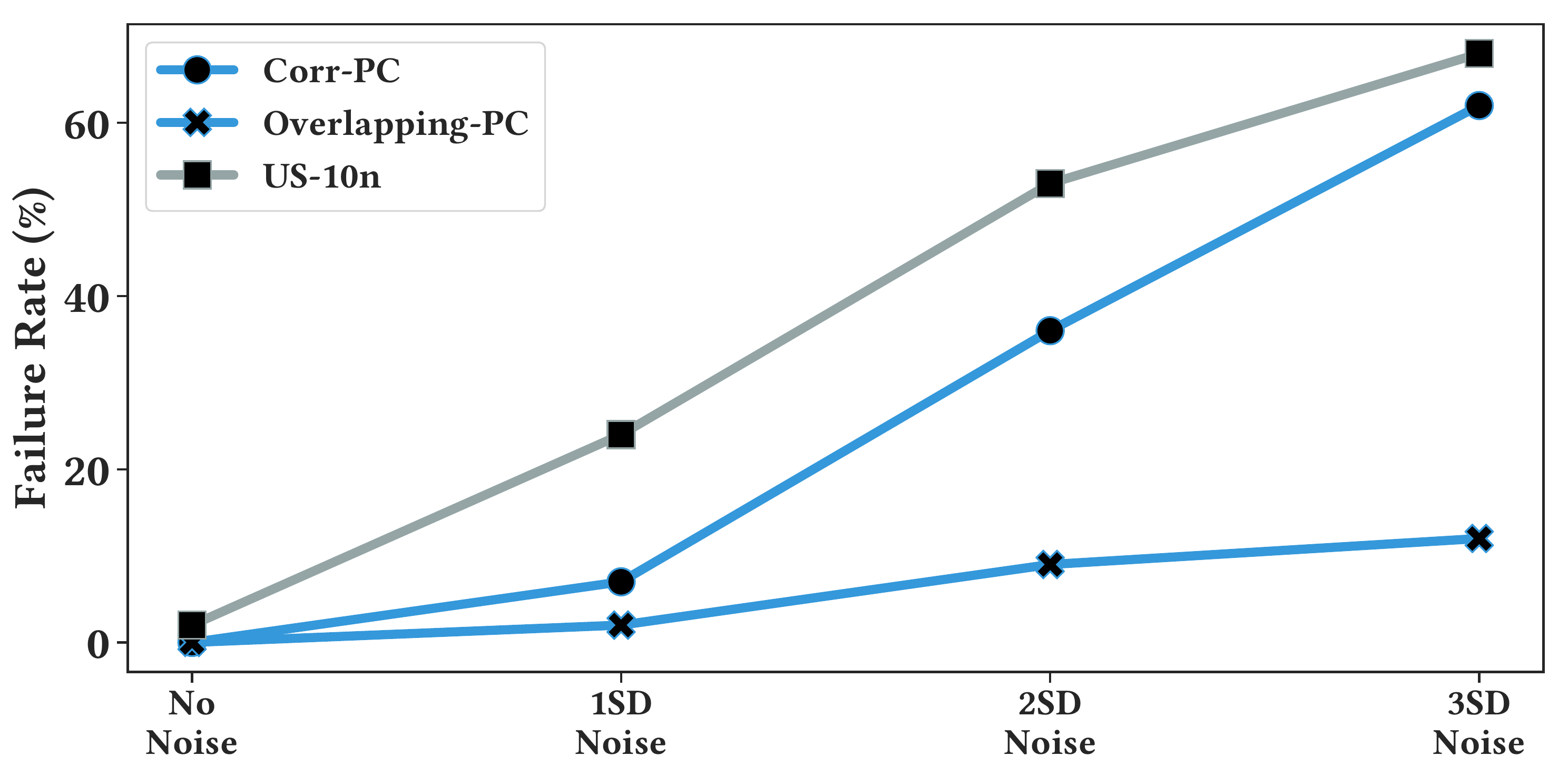}
    \caption{We investigate the sensitivity of Corr-PC, Overlapping-PC, and US-10n to different levels of noise. The failure rate of all approaches increases as we increase the noise level, but the PC baselines especially Overlapping-PC are more tolerable to the same level of noise.}
    \label{noisy-partition}
\end{figure}

\subsection{Scalability}\label{exp-scaling}
In this subsection, we evaluate the scalability of PCs and optimization techniques. As mentioned in earlier sections, the complete process of solving a PC problem including two parts. First, we need to perform cell decomposition to find out which cells are valid and second, we need to formalize a MILP problem with the valid cells that can be solved to find the optimal bound. 
The vanilla version described above is intractable because the number of sub-problems we need to solve in the cell decomposition phase is exponential to the number of PCs.

We presented a number of optimizations in the paper to improve this time.
We will show that naive processing of the PCs leads to impractical running time.
We generate 20 random PCs that are very significantly overlapping.
Figure \ref{dfs_tuning} plots the number of cells evaluated as well as the run time of the process.
The naive algorithm evaluates the SAT solver on more than 1000x more cells than our optimized approach.

Since cell decomposition is really the most expensive step
We can prune and save about 99.9\% of the solving time by using DFS (early termination for cell decomposition) and the rewriting heuristic.
Without these optimizations, PCs are simply impractical at large scales. 

\begin{figure}[t]
    \centering
    \includegraphics[width=0.5\linewidth]{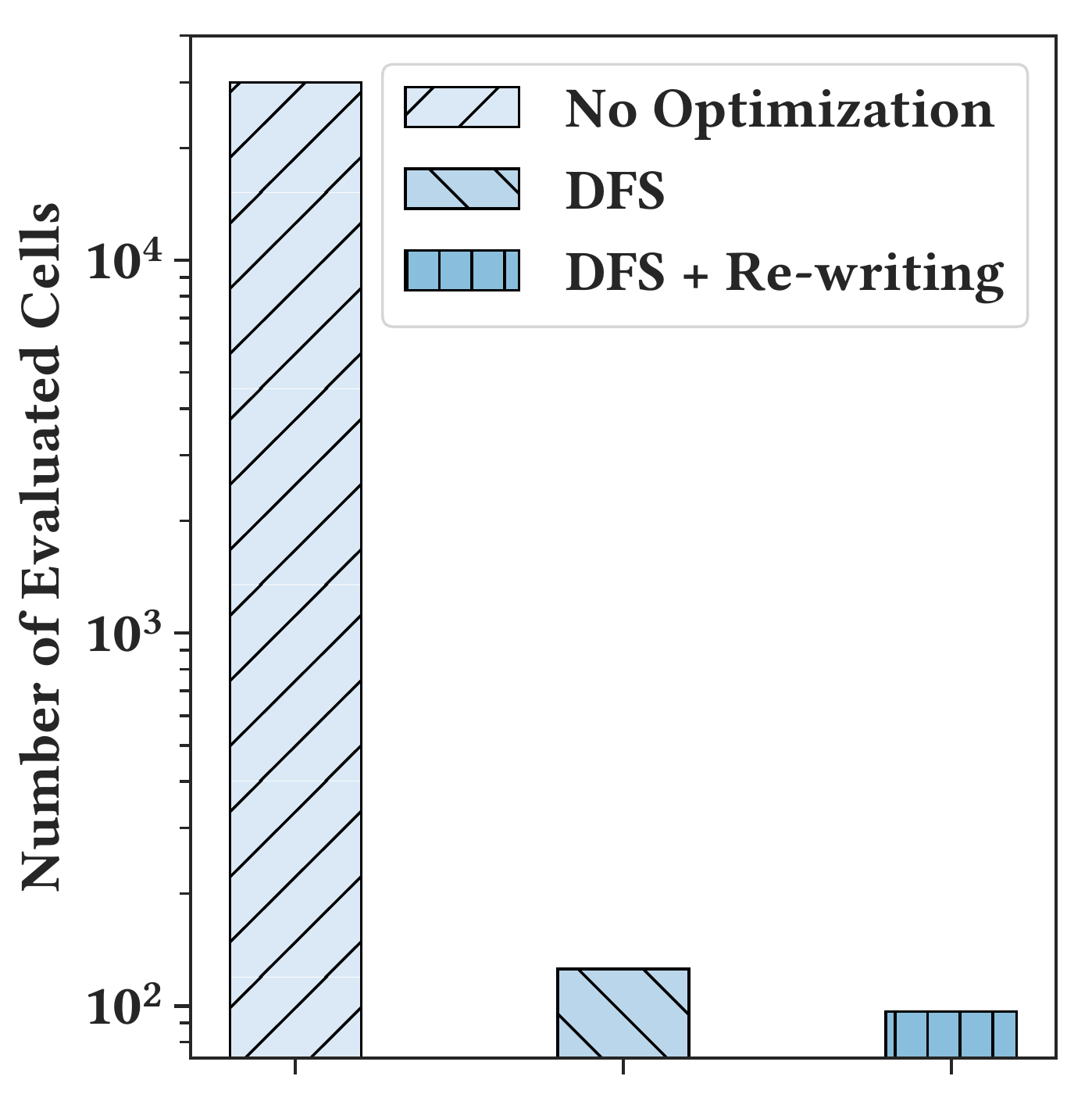}
    \caption{Our optimizations reduce the number of cells evaluated during the cell decomposition phase by over a 1000x.}
    \label{dfs_tuning}
\end{figure}

\subsubsection{Non-Overlapping PCs Scale Effectively}
PCs can be solved significantly faster in the special case of partitioned PCs (non-overlapping).
The process of answering a query with PC partitions is much simpler than using overlapping PCs. Because partitions are disjoint with each other, we can skip the cell decomposition, and the optimization problem can be solved by a greedy algorithm. As shown in Figure \ref{partition-scalability}, the average time cost to solve one query with a partition of size 2000 is 50ms, and the time cost is linear to the partition size. 
We can scale up to 1000s of PCs in this case.

\begin{figure}[t]
    \centering
    \includegraphics[width=1\linewidth]{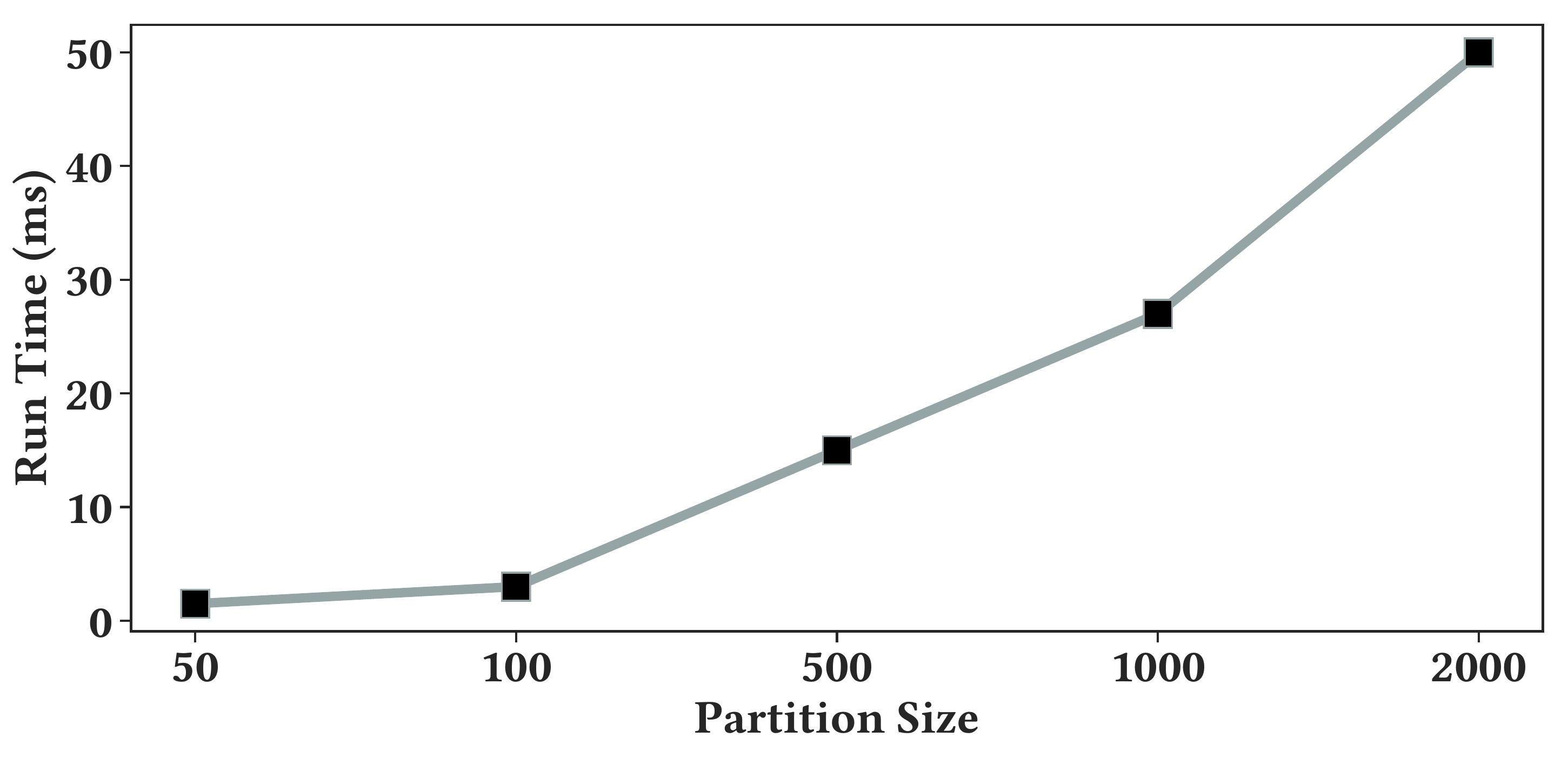}
    \caption{The run time needed to solve one query using partition of different size.}
    \label{partition-scalability}
\end{figure}


\subsection{Handling MIN, MAX, and AVG Queries}
As mentioned in earlier sections, besides COUNT and SUM queries, PCs can also handle MIN, MAX and AVG queries. In this experiment, we use the Intel Wireless dataset for demonstration, we partition the dataset on \texttt{DeviceID} and \texttt{Time}. For each type of query, we randomly generate 1000 queries and use PC to solve them. Results can be found in Figure \ref{mma}.
First, note how PC can always generate the optimal bound for MIN and MAX queries.
PCs are a very good representation of the spread of the data, more so than a sample.

We show similar performance for AVG queries to the COUNT and SUM queries studied before.
AVG queries are an interesting case that we chose not to focus on.
While sampling is a very accurate estimate of AVG queries without predicates, with predicates the story becomes more complicated.
Since averages are normalized by the number of records that satisfy the predicate, you get a ``ratio of estimators'' problem and the estimate is not exactly unbiased.
So for small sample sizes, standard bounding approaches can have a high failure rate despite seemingly accurate average-case performance.

\begin{figure}[t]
    \centering
    \includegraphics[width=0.5\linewidth]{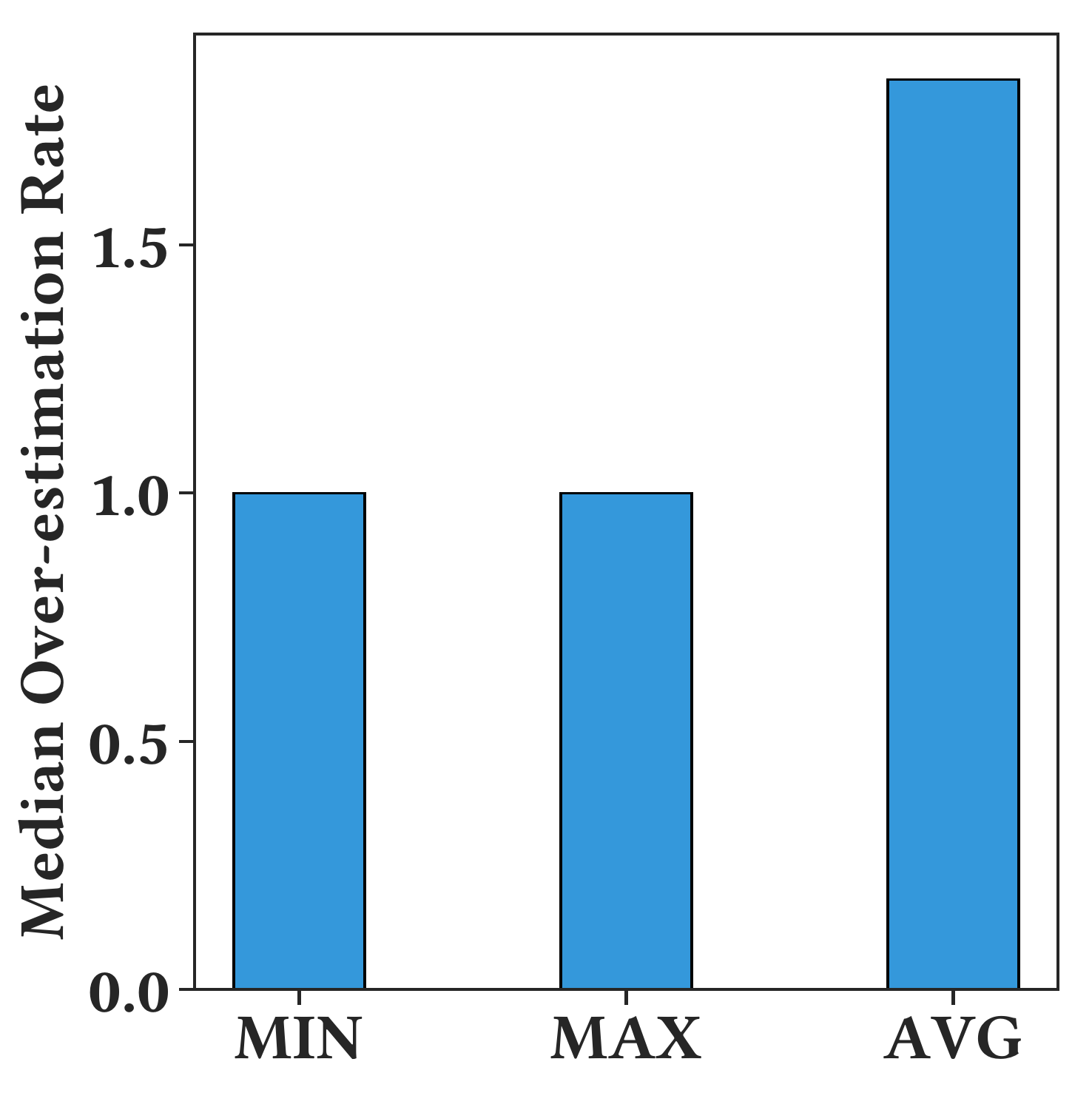}
    \caption{With PC, an optimal bound can be derived for MIN and MAX queries. PC also generates competitive result for AVG Queries.}
    \label{mma}
\end{figure}

\subsection{Additional Datasets}\label{exp-other}
We evaluate the accuracy of the framework using two other datasets. 

\subsubsection{Airbnb at New York City Dataset}
The Airbnb dataset \cite{airbnbnyc} contains open data from Airbnb listings and metrics in the city of New York, 2019. This dataset contains \reviewtwo{50 thousand rows}, 19 columns that describe different properties of listings like location (latitude, longitude), type of room, price, number of reviews, etc. Corr-PC and Rand-PC are defined as $n=1500$ constraints over \texttt{latitude} and \texttt{longitude}.

Figure \ref{nyab-lat-long} replicates the same experiment on a different dataset. This dataset is significantly skewed compared to the Intel Wireless dataset, so the estimates are naturally harder to produce.
As before, we find that well-designed PCs are just as tight as sampling-based bounds.
However, randomly chosen PCs are significantly looser (more than 10x).
PCs fail conservatively, a loose bound is still a bound, it might just not be that informative.
In skewed data such as this one, we advise that users design simple PCs that are more similar to histograms (partition the skewed attribute).

\begin{figure}[t]
    \centering
    \includegraphics[width=1\linewidth]{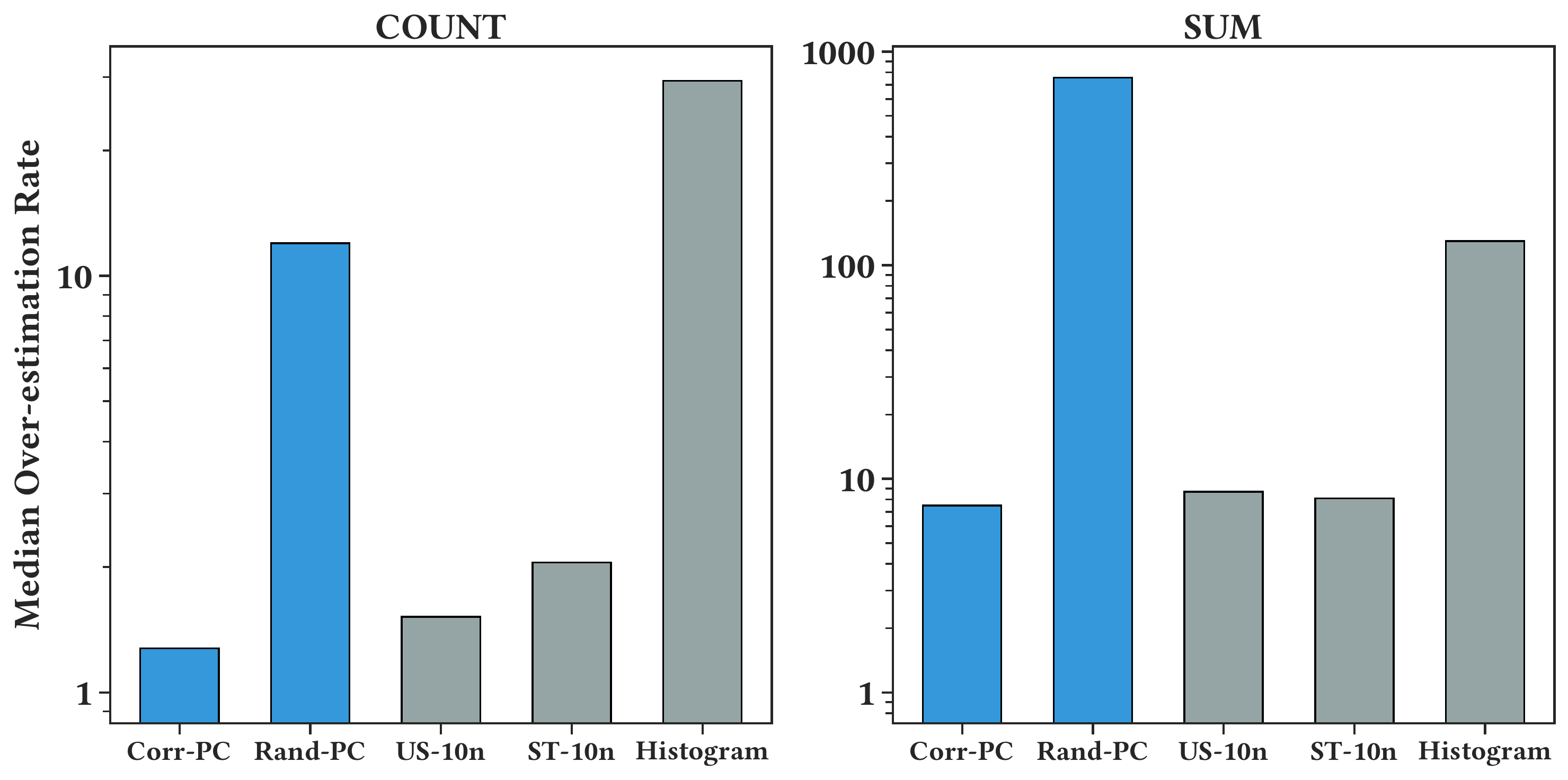}
    \caption{Baseline performance on \changed{1000} COUNT(*) and SUM queries with predicate attributes on \textit{Latitude} and \textit{Longitude} using the Airbnb NYC dataset.}
    \label{nyab-lat-long}
\end{figure}


\subsubsection{Border Crossing Dataset}
The Border Crossing dataset \cite{bordercrossing} from The Bureau of Transportation Statistics (BTS) summary statistics for inbound crossings at the U.S.-Canada and the U.S.-Mexico border at the port level. This dataset contains \reviewtwo{300 thousand rows}, 8 columns that describe the summary of border crossing (the type of vehicles and the count) that happen at a port (port code, port location, state, etc) on a specific date. We compare the hard bound baselines with PCs using different partitions with three groups of randomly generated queries.  Corr-PC and Rand-PC are defined as $n=1600$ constraints over \texttt{port} and \texttt{date}.

Results in Figure \ref{border-port-date} show results on another skewed dataset. As before, informed PCs are very accurate (in fact more accurate than sampling). Randomly chosen PCs over-estimate the result range by about 10x compared to the other approaches.
Again, the advantage of the PC framework is that unless the assumptions are violated, there are no random failures.
On this dataset, over 1000 queries, we observed one bound failure for the sampling approach.
This failure is included in the results.

\begin{figure}[t]
    \centering
    \includegraphics[width=1\linewidth]{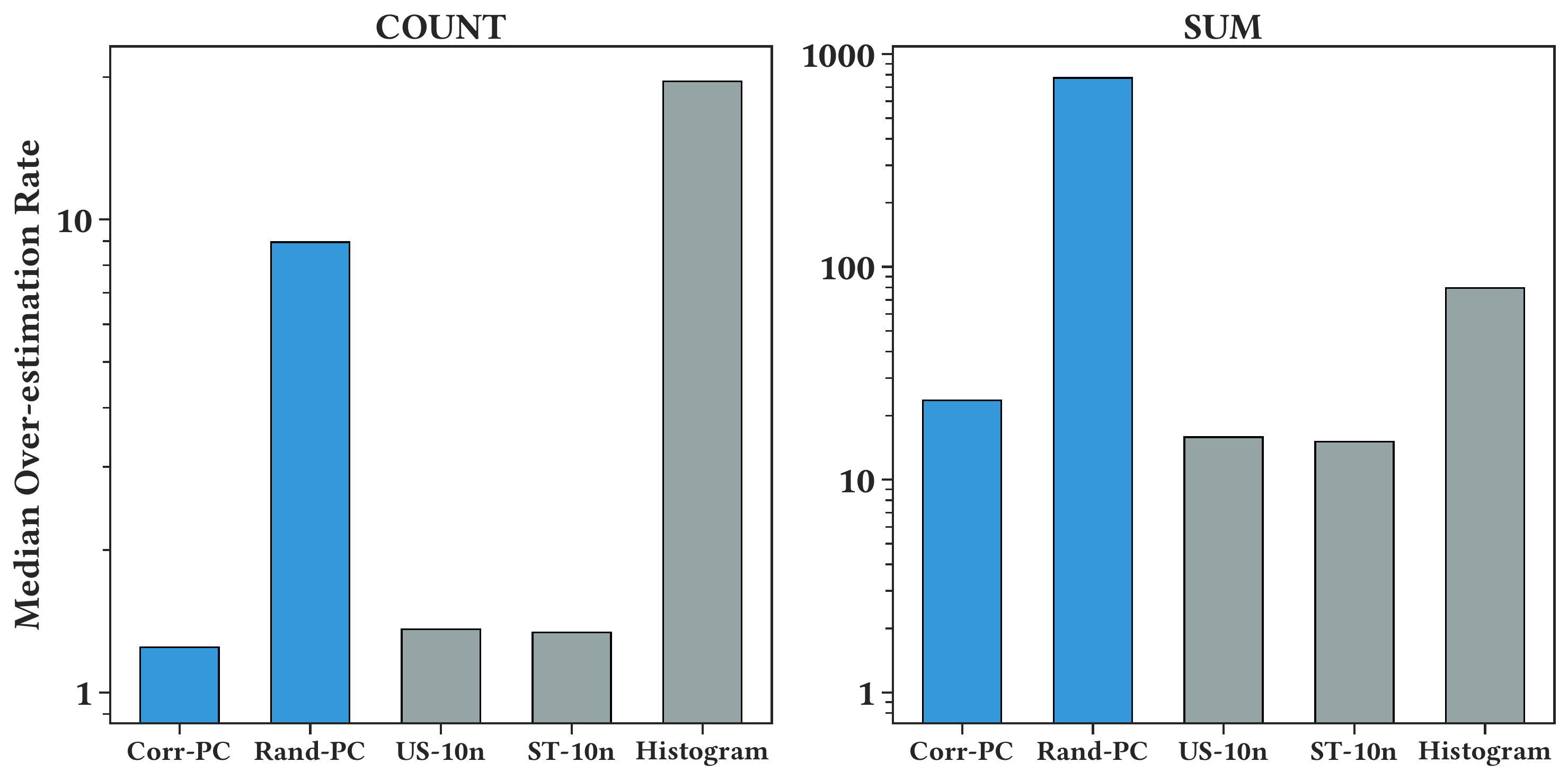}
    \caption{Baseline performance on \changed{1000} COUNT(*) and SUM queries with predicate attributes on \textit{Port} and \textit{Date} using the Border crossing dataset.}
    \label{border-port-date}
\end{figure}

\subsubsection{Join Datasets}
We also evaluate the PC framework on a number of synthetic join examples on randomly generated data.
The statistical approaches do not generalize well to estimates for queries with inner equal joins, and we found the bounds produced were too fallible for meaningful comparison.
To evaluate PCs on such queries, we compare to another class of bounding techniques that have been proposed in the privacy literature.
These bounds estimate how much a query might change for a single hypothetical point update.
Our insight connecting the bounding problem to worst-case optimal join results leads to far tighter bounds in those settings.
Johnson et~al. \cite{johnson2018towards} proposed a technique named elastic sensitivity that can bound the maximum difference between the query's result on two instances of a database. 

\noindent \emph{Counting Triangles.}
In this example, which is also studied by Johnson et al. \cite{johnson2018towards}, we analyze a query that is used to count triangles in a directed graph. 
 In Figure \ref{triangle} (TOP), we show the results of the two approaches on the counting triangle problem using randomly populated $edges$ tables of different sizes. And the results confirm that our approach drives a bound that is much tighter in this case---in fact by multiple orders of magnitude.

\begin{figure}[t]
    \centering
    \includegraphics[width=1\linewidth]{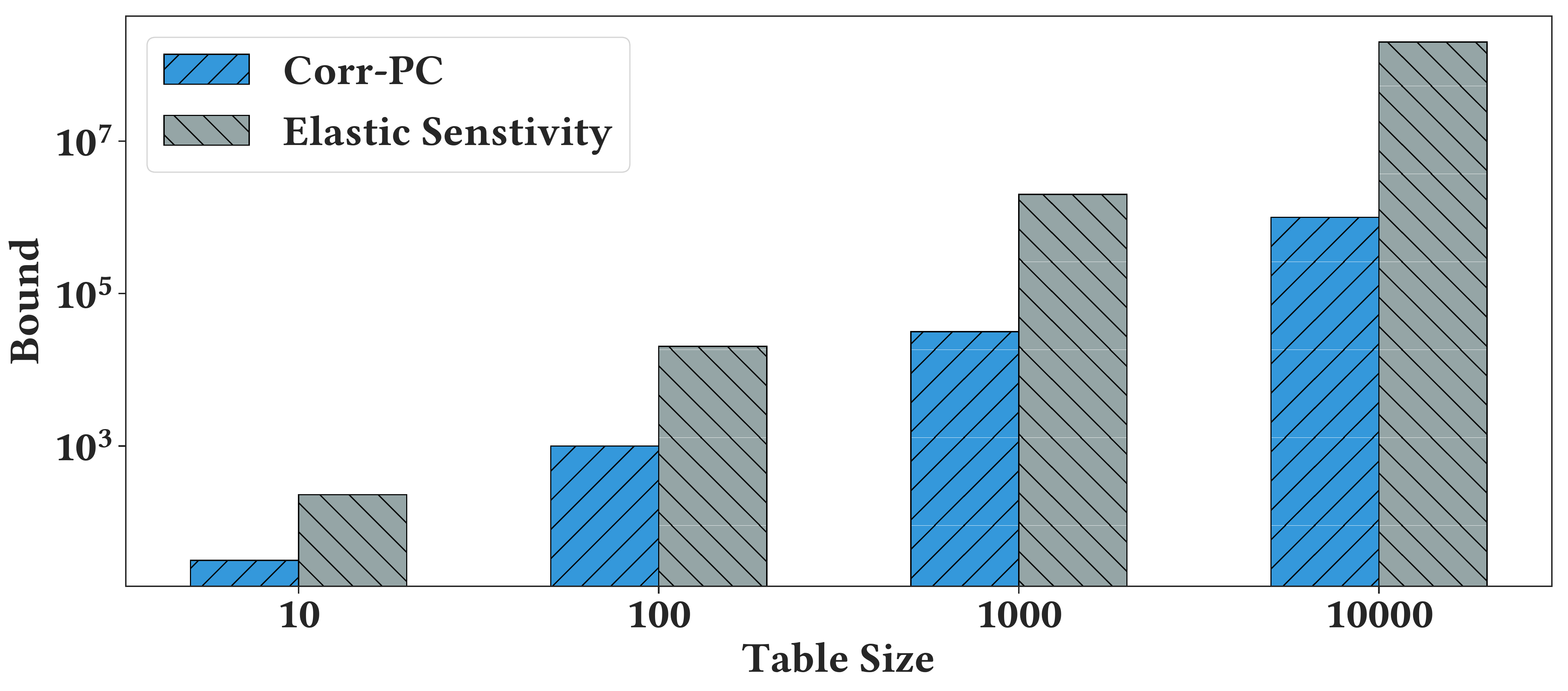}
     \includegraphics[width=1\linewidth]{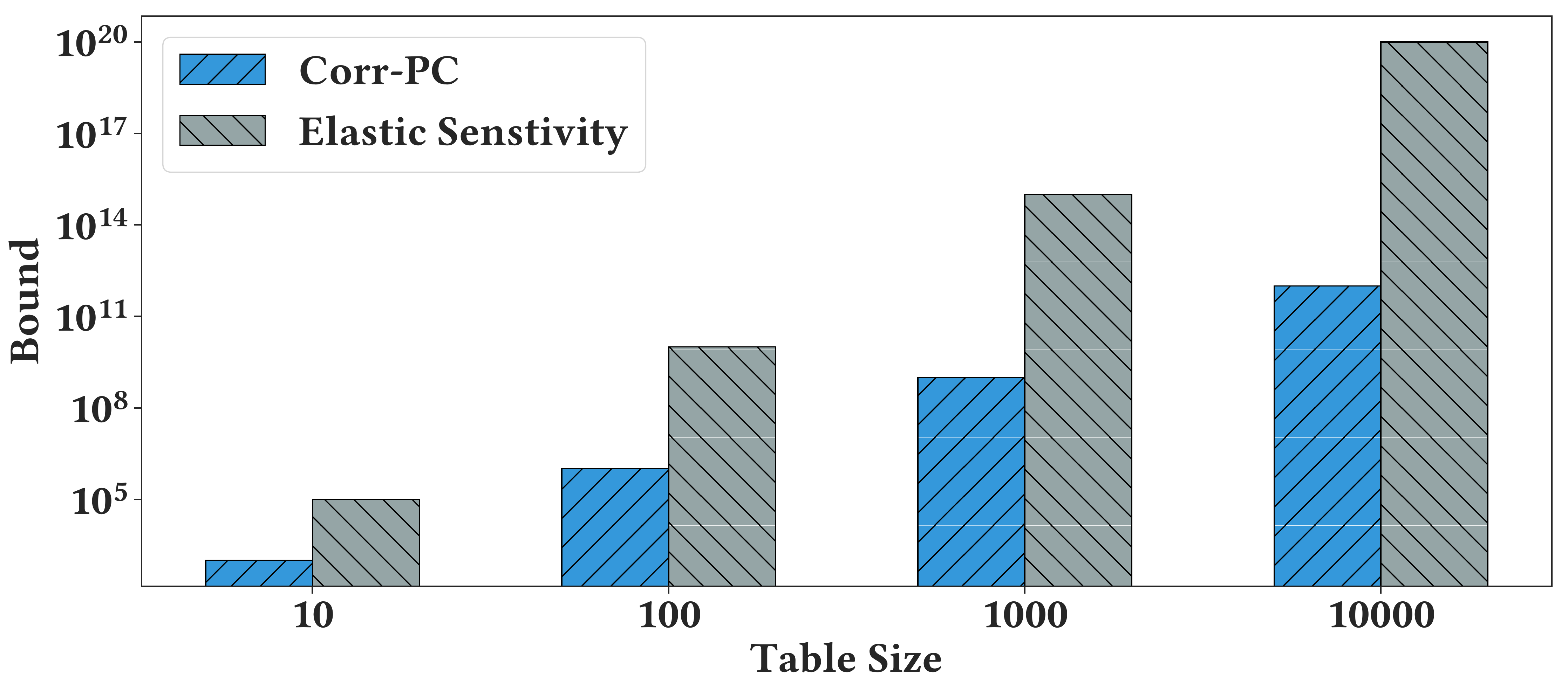}
    \caption{We compare the bound derived by our approach (Corr-PC) with state of the art baseline Elastic Sensitivity on the triangle counting problem of different table sizes (TOP) and an acyclic join (BOTTOM).}
    \label{triangle}
\end{figure}


\noindent \emph{Acyclic Joins.}
We also consider the following join query:
\[
R1(x_1, x_2) \bowtie R2(x_2, x_3)  ... \bowtie R5(x_5, x_{6})
\]
We generate 5 tables, each with $K$ rows and use the two approaches to evaluate the size of the join results. We vary the value of $K$ to understand how the bounds change accordingly. The results are shown in Figure \ref{triangle} (BOTTOM), we can see that elastic sensitivity always assumes the worst-case scenario thus generates the bound for a Cartesian product of the tables that is several magnitudes looser than our approach.

\begin{table*}[ht]\small
\begin{tabular}{llllllllllll}
\hline
\multicolumn{1}{c}{Dataset} & \multicolumn{1}{c}{Query} & \multicolumn{1}{c}{PredAttr} & \multicolumn{1}{c}{PC} & \multicolumn{1}{c}{Hist} & US-1p & US-10p & US-1n & \multicolumn{1}{c}{US-10n} & ST-1n & ST-10n & \multicolumn{1}{c}{Gen} \\ \hline
 &  & Time & 0 & 0 & 0 & 0 & 0 & 0 & 0 & 0 & 4 \\
 & COUNT(*) & DevID & 0 & 0 & 2 & 3 & 0 & 0 & 0 & 0 & 925 \\
Intel Wireless &  & DevID, Time & 0 & 0 & 12 & 3 & 0 & 0 & 0 & 0 & 591 \\
 &  & Time & 0 & 0 & 2 & 0 & 0 & 0 & 0 & 0 & 0 \\
 & SUM(Light) & DevID & 0 & 0 & 3 & 0 & 0 & 0 & 0 & 0 & 518 \\
 &  & DevID, Time & 0 & 0 & 24 & 4 & 8 & 2 & 9 & 1 & 205 \\ \hline
 &  & Latitude & 0 & 0 & 15 & 0 & 0 & 0 & 0 & 0 & 0 \\
 & COUNT(*) & Longitude & 0 & 0 & 36 & 0 & 0 & 0 & 0 & 0 & 0 \\
 &  & Lat, Lon & 0 & 0 & 39 & 0 & 0 & 0 & 0 & 0 & 0 \\
Airbnb@NYC &  & Latitude & 0 & 0 & 16 & 0 & 13 & 0 & 11 & 6 & 14 \\
 & SUM(Price) & Longitude & 0 & 0 & 36 & 0 & 35 & 0 & 24 & 16 & 3 \\
 &  & Lat, Lon & 0 & 0 & 45 & 19 & 39 & 0 & 39 & 13 & 17 \\ \hline
 &  & Port & 0 & 0 & 1 & 1 & 0 & 0 & 0 & 0 & 531 \\
 & COUNT(*) & Date & 0 & 0 & 1 & 0 & 0 & 0 & 0 & 0 & 744 \\
 &  & Port, Date & 0 & 0 & 13 & 1 & 0 & 0 & 0 & 0 & 362 \\
Border Cross &  & Port & 0 & 0 & 173 & 1 & 1 & 1 & 1 & 1 & 0 \\
 & SUM(Value) & Date & 0 & 0 & 33 & 1 & 1 & 0 & 0 & 0 & 0 \\
 &  & Port, Date & 0 & 0 & 192 & 12 & 20 & 0 & 15 & 0 & 0 \\ \hline
\end{tabular}

\caption{Over a \changed{1000} randomly chosen predicates, we record the number of failure events of different ``error bound'' frameworks. A failure event is one where an observed outcome is outside the range of the returned bounds. For a 99\% confidence interval, as used in our experiments, one would expect a 1\% failure rate for the sampling-based frameworks---but is significantly higher in practice for small skewed datasets. PCs, and as a special case Histograms, are guaranteed not to fail if the assumptions are satisfied.}
\label{failure-rate-tbl}
\end{table*}

\subsection{Probabilistic Confidence Intervals are Often Unreliable on Real Data}
Table \ref{failure-rate-tbl} presents different techniques and their ``failure rate'' over 1000 queries, which is the number of queries for which the true value exceeded what was produced in a bound.
The most common technique by far is to rely on the Central Limit Theorem (US-1p, US-10p).
Estimating this standard error from a sample is often far more unreliable than one would normally expect.
We use a 99\% confidence interval for a CLT bound given $N$ samples and $10N$ samples, and observe that the failure rate is far higher than 1\%.
In this missing data setting, a small number of example rows fail to accurately capture the ``spread'' of a distribution.

Next, we can make the sample-based confidence intervals a much more conservative non-parametric model (US-1n, US-10n), which holds under milder assumptions. Such a bound relies on an estimate of the min and max values and not an accurate estimate of the standard error.
Predictably, this performs much better than the CLT approach. 
However, as we can see in the table, non-parametric bound baselines still fail more often than one would expect over \changed{1000} queries.
Small samples and selective queries pose a fundamental challenge to these approaches.
Stratified samples do not solve this problem either. While, they cover the space more evenly, for any given strata, they can have a high failure rate. 

One could intuitively fix this problem by annotating the strata in a stratified sample with metadata that accurately depicts min and max values.
This is exactly the definition of PCs.
The PC technique and Histograms always generate hard bounds for queries because for the same number of ``bits'' of information they capture the entire spread of values much more accurately.
For the purposes of bounding, the example tuples provided by a sample are not as useful as the ranges.

Finally, we use the generative approach to model the joint data distribution.
We draw samples from this model and use that to produce a confidence interval.
Such an approach works very well on some datasets/queries but not others.
These experiments illustrate how important a guaranteed ``0 failure rate'' is for real-world decision making. Statistical confidence intervals can give a false sense of security in real-world data.

\section{Related Work}
\label{sec:related}
The overarching challenge addressed in the PC framework is related to the concept of ``reverse data management'' proposed by Meliou et al.~\citep{meliou2011reverse, meliou2012tiresias}. 
Meliou et al. argue that as data grow in complexity, analysts will increasingly want to know not what their data currently says but what changes have to happen to the dataset to force a certain outcome.
Such \emph{how-to} analyses are useful in debugging, understanding sensitivity, as well as planning for future data.
Meliou et al. build on a long line of \emph{what-if} analysis and data provenance research, which study simulating hypothetical updates to a database and understanding how query results might change~~\citep{deutch2013caravan, buneman2001and}.
While we address similar concerns to this line of work in spirit, our focus on aggregate queries and confidence intervals leads to a very different set of technical contributions.
The PC framework should be evaluated much more like a synopsis data structure than a data provenance reasoning systems.

Therefore, our experiments largely focus on evaluations against other synopsis structures and how to extract confidence intervals from them~\citep{cormode2011synopses}. While Approximate Query Processing (AQP) has been studied for several decades~\citep{Olston:2001:APS:375663.375710}, it is well-known that the confidence intervals produced can be hard to interpret~\citep{kennedy:2019:ifds:analyzing}.
This is because estimating the spread of high dimensional data from a small sample is fundamentally hard, and the most commonly used Central-Limit Theorem-based confidence intervals rely on estimated sample variance.
Especially for selective queries, these estimates can be highly fallible---a 95\% confidence interval may ``fail'' significantly more than 5\% of the time~\citep{agarwal2014knowing}. 
Unfortunately, as confidence intervals become more conservative, e.g., using more general statistical bounding techniques, their utility drops~\citep{hellerstein1997online}.
In a sense, our optimization algorithm automatically navigates this trade-off. 
The algorithm optimizes the tightest bound given the available information in the form of PCs.
We interpret PCs as generalized histograms with overlapping buckets and uncertain bucket counts.
Despite these differences with AQP, we do believe that the connections between uncertainty estimation and dirty data (like missing rows) are under-studied~\citep{krishnan2015stale, krishnan2015sampleclean}.
We also believe that in future work mixed systems with both PCs and samples can have the best of both worlds, e.g., augmenting Quickr with PC estimation~\citep{kandula2016quickr}.

Deterministic alternatives to AQP have been studied in some prior work.
Potti et al. propose a paradigm called DAQ \citep{potti2015daq} that does reason about hard ranges instead of confidence intervals.
DAQ models uncertainty at relation-level instead of predicate-level like in PCs and DAQ does not handle cardinality variation. 
In the limited scenario of windowed queries over time-series data, deterministic bounds have been studied~\citep{brito2017efficient}.
The technical challenge arises with overlapping constraints and complex query structures (like join conditions and arbitrary predicates).
Similarly, we believe that classical variance reduction techniques for histograms could be useful for PC generation in future work~\citep{poosala1996improved}, since 
histograms are a dense 1-D special case of our work.

\onethree{
c-tables are one of the classical approaches for representing missing data in a relation~\cite{imielinski1989incomplete}.
Due to the frequency constraints in Predicate-Constraint sets, we can represent cases that go beyond the typical closed-world assumption (CWA) is required in c-tables, where all records are known in advance and null cells are specifically annotated.}
There is also recent work that studies missing rows from databases. m-tables study variable cardinality representations to go beyond the CWA. In m-tables, cardinality constraints are specified per-relation. We specify frequency constraints per predicate. However, Sundarmurthy et al. ~\cite{sundarmurthy2017m} do not consider the problem of aggregate query processing on uncertain relations. There is similarly related work that studies intentionally withholding partitioned data for improved approximate query performance~\cite{Sundarmurthy:2018:EDP:3206333.3206337}. We believe that the novelty of our framework is the efficient estimation of aggregate query confidence intervals.
Similarly, the work by Burdik et al. is highly related where they study databases with certain ``imprecise'' regions instead of realized tuples~\citep{burdick2006efficient}. And the approach proposed by Cai et al. \citep{cai2019pessimistic} based on random hash partitions can only handle cardinality estimations over joins.
Cai et al. highlight many of the core challenges but fails to produce confidence intervals or handle inner-equality join queries optimally like our framework.

It is important to note, that our objective is not to build the most expressive language to represent uncertain data but rather one that we can pragmatically use to bound aggregate queries.

The privacy literature has studied a version of this problem: bounding aggregate queries on uncertain data~\citep{zhang2007aggregate, johnson2018towards}. In fact, Zhang et al. can be seen as solving the partitioned version of our problem~\citep{zhang2007aggregate}. However, they do not need to consider the overlapping case and joins in the way that our work does.

\section{Conclusion}
We proposed a framework that can produce automatic contingency analysis, i.e., the range of values an aggregate SQL query could take, under formal constraints describing the variation and frequency of missing data tuples.
There are several interesting avenues for future work.
First, we are interested in studying these constraints in detail to model dirty or corrupted data.
Rather than considering completely missing or dirty rows, we want to consider rows with some good and some faulty information.
From a statistical inference perspective, this new problem statement likely constitutes a middle ground between sampling and Predicate-Constraints.
Second, we would like to further understand the robustness properties of result ranges computed by Predicate-Constraints as well as other techniques.
Understanding when result ranges are meaningful for real-world analytics will be an interesting study.
Finally, we would like to extend the Predicate-Constraint framework to be more expressive and handle a broader set of queries.

\bibliographystyle{abbrv}
\bibliography{refs}

\end{document}